%% file: NF-decision-diagrams-arx.tex
\documentclass[]{llncs}

\usepackage{amssymb}
\usepackage{amsfonts}
\usepackage{mathtools}
\usepackage{calc}
\usepackage{hyperref}
\usepackage{graphicx}
\usepackage{multicol}
\usepackage{amsmath}
\usepackage{stmaryrd}

\usepackage{tikz}
\usepackage{pgfplots}
\usetikzlibrary{decorations.markings}
\usetikzlibrary{backgrounds,folding}
\usetikzlibrary{shapes.geometric, shapes.misc}
\pgfdeclarelayer{edgelayer}
\pgfdeclarelayer{nodelayer}
\pgfsetlayers{background,edgelayer,nodelayer,main}
\tikzstyle{every picture}=[baseline=-0.25em]
\tikzstyle{none}=[inner sep=0mm]
\tikzstyle{black dot}=[inner sep=0.7mm,minimum width=0pt,minimum height=0pt,fill=black,draw=black,shape=circle]
\tikzstyle{dot}=[black dot]
\tikzstyle{white dot}=[dot,fill=white]
\tikzstyle{box}=[rectangle,fill=white,draw=black, font=\scriptsize, inner sep=2pt]
\tikzstyle{box-no-outline}=[rectangle, draw=white, fill=white, inner sep=2pt]
\tikzstyle{comonoid}=[draw=black, trapezium, trapezium left angle=70, trapezium right angle=70, fill=black!80, inner sep=2.2pt]

\pgfdeclarelayer{edgelayer}
\pgfdeclarelayer{nodelayer}
\pgfsetlayers{background,edgelayer,nodelayer,main}
\tikzstyle{every loop}=[]

\usepackage{etoolbox}

\newtoggle{extern}
\toggletrue{extern}
\togglefalse{extern}

\newcommand{\tikzfig}[1]{
\input{./figures/#1.tikz}
}

\def\fig{}

\iftoggle{extern}{
	\usetikzlibrary{external}
	\tikzexternalize[prefix=tikzfigs/]
	
	\renewcommand{\tikzfig}[1]{
		\tikzsetnextfilename{#1}
		
\input{./figures/#1.tikz}
}
	
}{
	
}

\newcommand{\eq}[2][~~]{
#1
\underset{\substack{#2}}{=}
#1
}

\newcommand{\fit}[1]{\resizebox{\columnwidth}{!}{#1}}

\newcommand{\interp}[1]{\left\llbracket #1 \right\rrbracket}

\newcommand{\annoted}[3]{{\scriptstyle #1}\left\lbrace\mathrlap{\phantom{#3}}\right.\overbrace{#3}^{#2}}

\newcommand{\bra}[1]{\ensuremath{\left\langle #1 \right|}}
\newcommand{\ket}[1]{\ensuremath{\left|  #1 \right\rangle}}

\newcommand{\cat}[1]{\mathbf{#1}}

\allowdisplaybreaks

\title{Quantum Multiple-Valued Decision Diagrams in Graphical Calculi}

\begin{document}

\author{Renaud Vilmart}

\institute{Université Paris-Saclay, ENS Paris-Saclay, Inria, CNRS, LMF, 91190, Gif-sur-Yvette, France}

\maketitle

\begin{abstract}
Graphical calculi such as the ZH-calculus are powerful tools in the study and analysis of quantum processes, with links to other models of quantum computation such as quantum circuits, measurement-based computing, etc.

A somewhat compact but systematic way to describe a quantum process is through the use of quantum multiple-valued decision diagrams (QMDDs), which have already been used for the synthesis of quantum circuits as well as for verification.

We show in this paper how to turn a QMDD into an equivalent ZH-diagram, and vice-versa, and show how reducing a QMDD translates in the ZH-Calculus, hence allowing tools from one formalism to be used into the other.
\end{abstract}

\section{Introduction}

Graphical calculi for quantum computation such as the ZX-Calculus \cite{interacting}, the ZW-Calculus \cite{ghz-w} and the ZH-Calculus \cite{ZH} are powerful yet intuitive tools for the design and analysis of quantum processes. They have already been succesfully applied to the study of measurement-based quantum computing \cite{mbqc}, error correction through the operations of lattice surgery on surface codes \cite{dBDHP19,de2017zx}, as well as for the optimisation of quantum circuits \cite{Backens2020ThereAB,beaudrap2020fast,Kissinger2020reducing}.
Their strong links with ``sums-over-paths'' \cite{SOP,LvdWK,SOP-Clifford}, as well as their respective complete equational theories \cite{Backens2020ThereAB,HNW,JPV20,euler-zx}, make them good candidates for automated verification \cite{chancellor2016coherent,duncan2014verifying,MSc.Hillebrand}.

An important question, whose answer benefits a lot of these different aspects, is the one of synthesis. Given a description of a quantum process, how do we turn it into a ZX-diagram? This all depends on the provided description. It was already shown how to efficiently get a diagram from quantum circuits \cite{Backens2020ThereAB}, from a measurement-based process \cite{mbqc}, from a sequence of lattice surgery operations \cite{de2017zx}, from ``sums-over-paths'' \cite{LvdWK}, or even from the whole matrix representation of the process \cite{ZXNormalForm}.
Although this last translation is efficient in the size of the matrix, the size of the matrix itself grows exponentially in the number of qubits, so few processes will actually be given in terms of their whole matrix.

The matrix representation however has an advantage: it is (essentially) unique. Two quantum operators are operationally the same if and only if their matrix representations are colinear. This is to be contrasted with all the different previous examples, where for instance two different quantum circuits may implement the same operator.

The form of the ZX-diagram obtained from a quantum state by \cite{ZXNormalForm} is that of a binary tree: a branching in the tree corresponds to a cut in half of the represented vector, while the leaves of the tree exactly correspond to the entries in the vector.
It is however possible to exploit redundancies in the entries of the vector, by merging similar subtrees. Doing so alters the notion of normal form by compacting it, whilst retaining its uniqueness property.
This can be done at the level of the ZX-diagram using its equational theory, and in particular some equality that is reminiscent to that of a bialgebra rule.
Doing this from a proper tree on the other hand gives rise to a quantum version of a decision diagram, which has already been introduced in \cite{MT06}. The so-called quantum multiple-valued decision diagrams (QMDDs) \cite{NWMTD16} have since then been used to synthesise quantum circuits \cite{Niemann2020synthesis} or to perform verification of quantum programs \cite{BW20,burgholzer2020verifying}.

We hence aim in this paper at showing the links between the aforementioned graphical calculi and QMDDs. We in particular show how to translate from one formalism to the other, and how the reduction of a QMDD translates in the graphical languages.
As a consequence, tools developped in one formalism may be transported and used in the other. Additionally, this result together with the aforementionned results in the graphical languages, relates the QMDDs to measurement-based computation, lattice surgery operations, ``sums-over-paths'', etc.

In Section \ref{sec:ZH}, we present the ZH-calculus, the graphical language we will use in this paper for convenience. We then present in Section \ref{sec:SQMDD} the quantum multiple-valued decision diagrams. In Section \ref{sec:SQMDD-to-ZH}, we show how to turn a QMDD into a ZH-diagram that represents the same quantum operator. In Section \ref{sec:SQMDD-form}, we show an algorithm to turn a ZH-diagram into QMDD form, which can be used to get the QMDD description of any ZH-diagram.

\section{The ZH-Calculus}
\label{sec:ZH}

We aim in this paper at showing links between quantum multiple-valued decision diagrams and graphical languages for quantum computing: ZX, ZW and ZH. Since any of the three languages can be translated in the other two \cite{ZH,HNW,JPV}, we may simply choose one. It so happens that the closest to QMDDs we have is the ZH-Calculus. We hence present here this language.

\subsection{ZH-Diagrams}

A ZH-diagram $D:k\to \ell$ with $k$ inputs and $\ell$ outputs is generated by:\\
\begin{minipage}{0.55\columnwidth}
\begin{itemize}
\item $Z_m^n:n\to m::~~\tikzfig{Z-spider}~~$ called Z-spiders
\item $H_m^n(r):n\to m::~~\tikzfig{H-spider}~~$ called H-spiders
\end{itemize}
\end{minipage}\hfill
\begin{minipage}{0.43\columnwidth}
\begin{itemize}
\item $id:1\to1::~~\tikzfig{id}$
\item $\sigma:2\to 2::~~\tikzfig{swap}~~$
\item $\eta:0\to 2::~~\tikzfig{cap}~~$
\item $\epsilon:2\to 0::~~\tikzfig{cup}~~$
\end{itemize}
\end{minipage}\\
where $n,m\in\mathbb N$ and $r\in\mathbb C$. In the following, we may write H-spiders with no parameter, in which case, the implied parameter is $-1$ by convention.\\
Diagrams can then be composed either sequentially: $\tikzfig{sequence}$ (if the number of output of the top diagram matches the number of inputs of the bottom one), or in parallel: $\tikzfig{tensor}$.\\
It is customary to define the additional two ``X-spiders'':
$$\tikzfig{X-spider}\qquad\text{ and }\qquad\tikzfig{X-spider-neg}$$

ZH-diagrams can be understood as quantum operators thanks to the \emph{standard interpretation} $\interp{.}$ which maps any ZH-diagram $D:n\to m$ to a complex matrix $\interp{D}\in\mathbb C^{2^m}\times \mathbb C^{2^n}$, and which is inductively defined as:\\
\phantom{.}\hfill$\interp{\tikzfig{sequence}}=\interp{D_2}\circ\interp{D_1}$\hfill$\interp{\tikzfig{tensor}}=\interp{D_1}\otimes\interp{D_2}$\hfill~\\
{\renewcommand*{\arraystretch}{0.8}
\setlength{\arraycolsep}{2pt}
\phantom{.}\hfill$\interp{\tikzfig{Z-spider}}=\annoted{2^m}{2^n}{\begin{pmatrix}1&0&\cdots&\cdots&0\\[-0.5em]0&0&&&\vdots\\[-0.5em]\vdots&&\ddots&&\vdots\\[-0.5em]\vdots&&&0&0\\0&\cdots&\cdots&0&1\end{pmatrix}}$\hfill$\interp{\tikzfig{H-spider}}=\annoted{2^m}{2^n}{\begin{pmatrix}1&\cdots&\cdots&1\\[-0.5em]\vdots&\ddots&&\vdots\\[-0.5em]\vdots&&1&1\\1&\cdots&1&r\end{pmatrix}}$\hfill~}\\
\phantom{.}\hfill$\interp{~\tikzfig{id}~}=\begin{pmatrix}1&0\\0&1\end{pmatrix}$\hfill$\interp{\tikzfig{swap}}=\begin{pmatrix}1&0&0&0\\0&0&1&0\\0&1&0&0\\0&0&0&1\end{pmatrix}$\hfill$\interp{\tikzfig{cap}}=\interp{\tikzfig{cup}}^\dagger=\begin{pmatrix}1\\0\\0\\1\end{pmatrix}$\hfill~

The standard interpretation of the X-spiders can then be obtained by composition. We underline that 
$\interp{\tikzfig{ket-0}}=\ket0:=\begin{pmatrix}1\\0\end{pmatrix}$ and that
$\interp{\tikzfig{ket-1}}=\ket1:=\begin{pmatrix}0\\1\end{pmatrix}$.\\
We have used here the Dirac notion, were a quantum state i.e. a vector is denoted $\ket\psi$. We recall that $\bra\psi$ is defined as $\ket\psi^\dagger$ where $(.)^\dagger$ yields the transconjugate of a matrix.

The language is universal, i.e.~any quantum operator can be represented as a ZH-diagram~\cite{ZH}:
$$\forall f\in \mathbb C^{2^m}\times \mathbb C^{2^n},~\exists D\in\cat{ZH}(n,m),~~\interp{D}=f$$

An important result that we will use in the following is the fact that there is an isomorphism between $\cat{ZH}(n,m)$ and $\cat{ZH}(0,m+n)$, i.e.~any ZH-diagram $D:n\to m$ can be turned into a \emph{state} $D':0\to m+n$ in a reversible manner. This is called the \emph{map/state} duality \cite{interacting,Choi1975completely,Jamiolkowski1972linear}.\\
Graphically, this isomorphism is obtained by $\psi_{n,m}(D)=\tikzfig{map-state-duality-1-bis}$ and $\psi^{-1}_{n,m}(D')=\tikzfig{map-state-duality-2-bis}$, and the fact that this is indeed an isomorphism comes from the fact that:\\
\centerline{$\tikzfig{snake}\qquad\qquad\tikzfig{swap-involution}$}\\
Notice that this definition of the map/state duality differs from more usual ones by a rearranging of the wires. This is useful in the following to better relate state-QMDDs to proper QMDDs.

\subsection{Equational Theory}

The previous equalities constitute the first of a series of axioms that makes up an equational theory for the language. The axioms are summed up in figure \ref{fig:ZH-rules}.

\begin{figure}[!htb]
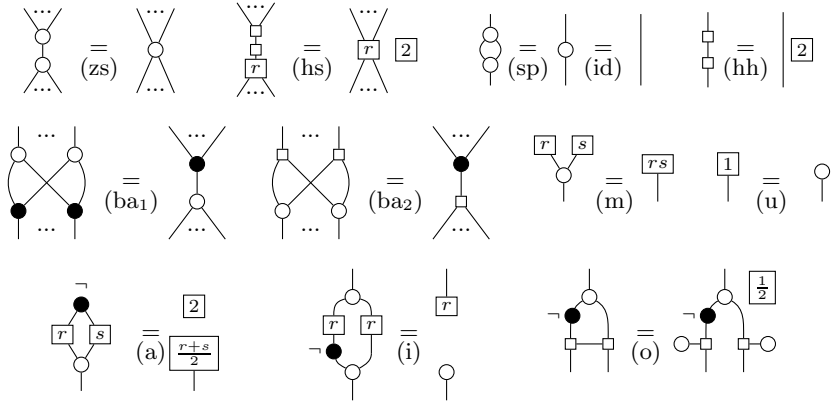

\phantom{.}\hfill$\tikzfig{rule-Z-spider}$\hfill$\tikzfig{rule-H-spider}$\hfill$\tikzfig{rule-identity}$\hfill$\tikzfig{rule-H-involution}$\hfill~\\[1em]
\phantom{.}\hfill$\tikzfig{rule-bialgebra}$\hfill$\tikzfig{rule-bialgebra-2}$\hfill$\tikzfig{rule-product}$\hfill$\tikzfig{rule-ket-plus}$\hfill~\\[1em]
\phantom{.}\hfill$\tikzfig{rule-average}$\hfill$\tikzfig{rule-inverse}$\hfill$\tikzfig{rule-O}$\hfill~
\caption{Rules of the ZH-Calculus}
\label{fig:ZH-rules}
\end{figure}
The previous equality is actually part of an implicit set of axioms of the language, aggregated under the paradigm ``only connectivity matters'', which states that all deformations of diagrams are allowed.

When we can turn a diagram $D_1$ into another diagram $D_2$ by a succession of the transformations in Figure \ref{fig:ZH-rules}, we usually write $\operatorname{ZH}\vdash D_1=D_2$, however, to keep things simple, we will abbreviate it as $D_1=D_2$ in this paper. 
This set of rules was proven to be sound and complete \cite{ZH}, that is:
$$\forall D_1,D_2\in \cat{ZH}(n,m),~~ \interp{D_1}=\interp{D_2}\iff \operatorname{ZH}\vdash D_1=D_2$$

A list of useful lemmas, proven in \cite{backens2021ZHcompleteness}, and that will be used in the proofs of the paper is given in Section \ref{sec:useful-lemmas} in the appendix.

In the following, since $\interp{\tikzfig{scalar-r}}=r$ and since it represents a global scalar, we may ``forget the box'', and write $\tikzfig{scalar-r}~\tikzfig{D}$ simply as $r~\tikzfig{D}$. Notice that thanks to Lemma \ref{lem:product-gen}, we have $\tikzfig{rule-product-0}$ , which means that following the convention, different global scalars will get multiplied. For instance, we have $r_1\cdot D_1\otimes r_2\cdot D_2 = (r_1r_2)\cdot (D_1\otimes D_2)$.

\subsection{New Constructions}

To make the link between decision diagrams and ZH-diagrams, we feel it is simpler to introduce two new constructions:\\
\def\fig{gadget-def}
$\input{./figures/\fig/\fig_00.tikz}~:=~\frac12\input{./figures/\fig/\fig_01.tikz}$ and 
\def\fig{monoid-def}
$\input{./figures/\fig/\fig_00.tikz}~:=~\input{./figures/\fig/\fig_01.tikz}\eq{\ref{lem:0-box}}\frac12~\input{./figures/\fig/\fig_02.tikz}$\\
of which we may compute the standard interpretation:\\
$\interp{\input{./figures/\fig/\fig_00.tikz}} = \begin{pmatrix}1&0&0&0\\0&1&1&0\end{pmatrix}$ and $\interp{\tikzfig{gadget-upward}} = \begin{pmatrix}1&0&1&0\\0&0&0&1\\0&1&0&0\\0&0&0&0\end{pmatrix}$.

$\input{./figures/\fig/\fig_00.tikz}$ corresponds to the diagram $\tikzfig{ZW-monoid}$ of the ZW-Calculus \cite{ghz-w}, it is very close to the so-called W-state. It can also by found in \cite{backens2021ZHcompleteness} in the context of the ZH-calculus. The interaction of this building block with classical states is given by:
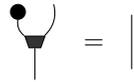
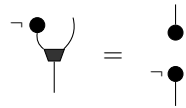
\begin{multicols}{2}
\begin{lemma}
\label{lem:ket-0-monoid}
\def\fig{monoid-ket0}
\begin{align*}
\input{./figures/\fig/\fig_00.tikz}
\eq{}\input{./figures/\fig/\fig_04.tikz}
\end{align*}
\end{lemma}
\begin{lemma}
\label{lem:ket-1-monoid}
\def\fig{monoid-ket1}
\begin{align*}
\input{./figures/\fig/\fig_00.tikz}
\eq{}\input{./figures/\fig/\fig_04.tikz}
\end{align*}
\end{lemma}
\end{multicols}
\begin{proof}
In appendix at page \pageref{prf:monoid-ket-0-1}.
\end{proof}

The pair $\left(\input{./figures/\fig/\fig_00.tikz},\tikzfig{ket-0}\right)$ can be seen as a commutative monoid, which means that, on top of Lemma \ref{lem:ket-0-monoid}, the following is true:
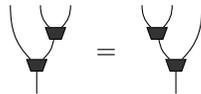
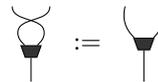
\begin{multicols}{2}
\begin{lemma}
\label{lem:monoid-associativity}
\def\fig{monoid-gen}
\begin{align*}
\tikzfig{monoid-associativity}
\end{align*}
\end{lemma}
\begin{lemma}
\label{lem:monoid-commutativity}
\begin{align*}
\tikzfig{monoid-commutative}
\end{align*}
\end{lemma}
\end{multicols}
\begin{proof}
Proof for associativity can be found in \cite{backens2021ZHcompleteness}, and commutativity is obvious from the definition.
\end{proof}

This allows us to use a generalised version of this diagram, with arbitrary number of inputs, defined inductively as follows:
\def\fig{monoid-gen-def}
$
\input{./figures/\fig/\fig_00.tikz}
~:=~\input{./figures/\fig/\fig_01.tikz}$ and $
\input{./figures/\fig/\fig_02.tikz}
~:=~\input{./figures/\fig/\fig_03.tikz}
$. 
In practice, in the following, we will not use the case $0$ inputs. We will however use the case $1$ extensively, which we will assume simplified:
\def\fig{monoid-gen-def}
$
\input{./figures/\fig/\fig_04.tikz}
~=~\input{./figures/\fig/\fig_05.tikz}
$

The second diagram is the ZH version of the gadget used in the normal forms of \cite{ZXNormalForm}, and it can be understood as follows:
$$\interp{\tikzfig{gadget-upward}} = \bra0\otimes\interp{~\def\fig{gadget-0-ctrld}\input{./figures/\fig/\fig_04.tikz}~}+\bra1\otimes\interp{~\def\fig{gadget-1-ctrld}\input{./figures/\fig/\fig_03.tikz}~}$$
so it is a diagram that builds $\def\fig{gadget-0-ctrld}\input{./figures/\fig/\fig_04.tikz}$ and whose leftmost wire controls whether or not the outputs are swapped. This can be checked diagrammatically:
\begin{multicols}{2}
\begin{lemma}
\label{lem:ket-0-gadget-ctrl}
\def\fig{gadget-0-ctrld}
\begin{align*}
\input{./figures/\fig/\fig_00.tikz}
\eq{}\input{./figures/\fig/\fig_04.tikz}
\end{align*}
\end{lemma}
\begin{lemma}
\label{lem:swapped-gadget-legs}
\def\fig{gadget-swap}
\begin{align*}
\input{./figures/\fig/\fig_00.tikz}
\eq{}\input{./figures/\fig/\fig_05.tikz}
\end{align*}
\end{lemma}
\begin{lemma}
\label{lem:ket-1-gadget-ctrl}
\def\fig{gadget-1-ctrld}
\begin{align*}
\input{./figures/\fig/\fig_00.tikz}
\eq{}\input{./figures/\fig/\fig_03.tikz}
\end{align*}
\end{lemma}
\begin{lemma}
\label{lem:ket-0-top-gadget}
\def\fig{gadget-ket0-top}
\begin{align*}
\input{./figures/\fig/\fig_00.tikz}
\eq{}\input{./figures/\fig/\fig_04.tikz}
\end{align*}
\end{lemma}
\end{multicols}

\begin{proof}
In appendix at page \pageref{prf:gadget-ket-0-1}.
\end{proof}

\section{Quantum Multiple-valued Decision Diagrams}
\label{sec:SQMDD}

Quantum multiple-valued decision diagrams (QMDD) were introduced to store quantum unitaries in a way that is efficient in certain cases, similarly to binary decision diagrams for representing decision problems. In the following, we use the map/state duality to turn every map into a state. As a consequence, QMDDs are adapted to ``state-QMDDs'' as follows:

\begin{definition}
A \emph{state-QMDD} (SQMDD) is a tuple $(s, V, u_0, t, H, h, E, \omega)$ where:
\begin{itemize}
\item $s\in \mathbb C$ is called the overall scalar
\item $V\neq\emptyset$ is a set of vertices
\item $u_0,t\in V$ are two distinguished vertices, called respectively \emph{root} and \emph{terminal node}
\item $u_0=t\iff V=\{t\}$, i.e.~$u_0$ and $t$ coincide only if $V$ only contains one vertex
\item $H\in\mathbb N$ is the \emph{height} of the SQMDD
\item $h: V\mapsto \{0,...,H\}$ maps each vertex to their height in the SQMDD
\item $h(u)=0\iff u=t$
\item $E:V\setminus\{t\}\to V^2$ maps any non-terminal vertex to a pair of vertices. These are the edges of the SQMDD
\item If $E(u)=(v_0,v_1)$ then $h(v_i)<h(u)$ for $i\in\{0,1\}$
\item $\forall v\in V\setminus\{u_0\}, \exists u\in V, v\in E(u)$, i.e. all vertices have at least one parent
\item $\omega:V\setminus\{t\}\to \mathbb C^2$ maps edges to complex weights
\end{itemize}
\end{definition}
Notice that the requirement on the heights of two endpoints of an edge also enforces the fact that an SQMDD is acyclical.

When drawing a SQMDD, it is relevant to set all the same-height nodes at the same height in the representation. It is also customary to omit writing weights of $1$, as well as the vertices' names (we write instead their height). We highlight the root by an incoming wire, and distinguish the terminal node by drawing it as a square, instead of a circle for the other nodes (following \cite{MT06}'s convention). Finally, we display the overall scalar as a weight on the incoming edge of the root, and if $h(u_0)=H$, we omit $H$, which we otherwise specify at the top of the SQMDD.

\begin{example}
\label{ex:QMDD}
The diagram:
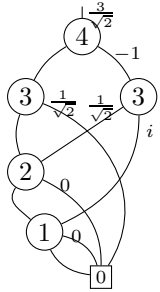
\[\def\fig{example-QMDD}\input{./figures/\fig/\fig_02.tikz}\]
is a graphical representation of an SQMDD.
\end{example}

\begin{remark}
SQMDDs can be seen as a more fine-grained version of proper QMDDs. Indeed, using the map/state duality $\psi_{n,m}$ defined above amounts to decomposing a QMDD vertex \tikzfig{QMDD-vertex} into \tikzfig{QMDD2SQMDD-vertex}. This for instance allows us to apply a swap $\tikzfig{swap}$ only on one side of a circuit/diagram, while the associated QMDD notion of variable reordering requires swapping the whole qubit i.e. apply a swap on both sides of the circuit. A similar presentation of QMDDs for states can be found in \cite{ZHMW20}.
\end{remark}

We can define two natural constructions from any SQMDD of height $\geq1$:
\begin{definition}
Let $\mathfrak D=(s, V, u_0, t, H, h, E, \omega)$ be a SQMDD with $H\geq1$. We denote $\ell(\mathfrak D)$ the diagram obtained from $\mathfrak D$ by:
\begin{itemize}
\item If $h(u_0)= H$:
\begin{itemize}
\item replacing the overall scalar $s$ by $s\cdot\pi_0\omega(u_0)$
\item removing all nodes (and subsequent edges) that cannot be reached by $\pi_0E(u_0)$
\item replacing the root $u_0$ by $\pi_0E(u_0)$
\end{itemize}
\item decreasing the height $H$ by $1$
\end{itemize}
Where $\pi_0$ and $\pi_1$ are the usual left and right projectors of a pair.\\
We can similarly build $r(\mathfrak D)$ by replacing $\pi_0$ by $\pi_1$ in the definition.
\end{definition}



\begin{example}
With $\mathfrak D$ the SQMDD defined in the previous example:\\
\def\fig{example-left-right-QMDD}$\ell(\mathfrak D)=\input{./figures/\fig/\fig_04.tikz}$ and $r(\mathfrak D)=\input{./figures/\fig/\fig_05.tikz}$
\end{example}


We can now use these constructions to understand SQMDDs as representations of quantum states.
\begin{definition}
For any SQMDD $\mathfrak D$, $\interp{\mathfrak D}$ is the (unnormalised) quantum state inductively defined as:
\begin{itemize}
\item $\interp{\raisebox{1ex}{$s$}\!\tikzfig{terminal-node}}=s=s\ket{}$
\item $\interp{\mathfrak D}=\ket0\otimes\interp{\ell(\mathfrak D)}+\ket1\otimes\interp{r(\mathfrak D)}$
\end{itemize}
(here $\ket{}$ is used to represent the vector $\begin{pmatrix}1\end{pmatrix}$ i.e. the canonical 0-qubit state).
\end{definition}

\begin{example}
\setcounter{MaxMatrixCols}{20}
With $\mathfrak D$ the diagram of example \ref{ex:QMDD}:
$$\interp{\mathfrak D}=\frac3{\sqrt2}\begin{pmatrix}1&0&0&0&\frac1{\sqrt2}&\frac1{\sqrt2}&\frac1{\sqrt2}&\frac1{\sqrt2}&-\frac1{\sqrt2}&0&0&0&-i&0&-i&0\end{pmatrix}^T$$
\end{example}

The definition of SQMDDs given here does not only differ from that of \cite{MT06} by the fact that we only consider states, but also by the fact that our definition is laxer. As a consequence, different SQMDDs can have the same interpretation. To address this problem, we can give a set of rewrite rules that will reduce the "size" of the SQMDD while preserving its interpretation. We give in Figure \ref{fig:QMDD-simplification} this set of rewrite rules, expressed graphically. Notice that the rules use the graphical notation to encompass transformations on the root and the overall scalar.

\begin{figure}[!htb]
\def\fig{QMDD-rule-scalar-distrib}$\input{./figures/\fig/\fig_00.tikz}\underset{\substack{a\neq0\\a\neq1}}{\to}\input{./figures/\fig/\fig_01.tikz}$\hfill\def\fig{QMDD-rule-scalar-distrib-2}$\input{./figures/\fig/\fig_00.tikz}\underset{\substack{b\neq1}}{\to}\input{./figures/\fig/\fig_01.tikz}$\hfill\tikzfig{QMDD-rule-0-to-terminal}
\\[1em]
\tikzfig{QMDD-rule-remove-coleaves}\hfill\tikzfig{QMDD-rule-variable-skipping}\hfill$\underset{\substack{\text{here the nodes with height $h'$}\\\text{and $h''$ can be the same node}}}{\tikzfig{QMDD-rule-regrouping-general}}$
\caption{Simplification rules for QMDDs}
\label{fig:QMDD-simplification}
\end{figure}

It is fairly easy to see that this rewriting terminates. Let us denote $\deg(t)$ the arity of $t$, i.e. the number of occurrences of $t$ in $E(V\setminus\{t\})$. Notice that $\deg(t)<2|V|$. For $u\in V\setminus\{t\}$, define:
\[\delta(u):=\begin{cases}0&\text{ if } \pi_0(\omega(u))=1\lor (\pi_0(\omega(u))=0\land \pi_1(\omega(u))\in\{0,1\})\\1&\text{ otherwise }\end{cases}\]
Consider now for any SQMDD the quantity:
$$\left(|V|, 2|V|-\deg(t), \sum\limits_{u\mid h(u)=1}\hspace*{-1em}\delta(u), ..., \sum\limits_{u\mid h(u)=H}\hspace*{-1em}\delta(u)\right)$$
and the lexicographical order over these. We can see that this quantity is reduced by any rewrite.

When none of these rewrite rules can be applied on an SQMDD, it is called irreducible. Notice that in an irreducible SQMDD, using the notions of \cite{MT06}, no non-terminal vertex is redundant, and all non-terminal vertices are normalised and unique. Hence, an irreducible SQMDD is what \cite{MT06} properly calls a QMDD, from which we get:

\begin{theorem}[\cite{MT06}]
\label{thm:QMDD-unique}
For any quantum state $\ket{\psi}\in\mathbb C^{2^n}$, there exists a unique irreducible SQMDD $\mathfrak D$ (of height $n$) such that $\interp{\mathfrak D}=\ket\psi$.
\end{theorem}

\section{From SQMDDs to ZH}
\label{sec:SQMDD-to-ZH}

From any SQMDD, it is possible to build a ZH-diagram whose interpretation will be the same, in a fairly straightforward manner, using the two syntactic sugars defined above. We denote $[.]^{\operatorname{ZH}}$ this map. We define it on every ``layer'' of an SQMDD, that is, on all the nodes of the same height $\leq H$. Such a layer is mapped to a ZH-diagram as follows:
\[\tikzfig{QMDD2ZH-layer}\]
This construction adds a wire to the left. It is the ``effective'' wire of the layer, the one that will constitute the output qubit in the quantum state. If there is no node of height $h$, the construction still add a wire on the left, disconnected from everything: $\tikzfig{ket-plus}$. Omitted in the previous definition is the mapping of the weights on wires, which are simply:
$\tikzfig{QMDD2ZH-weight}$, as well that of the terminal node, for which we have: $\tikzfig{QMDD2ZH-terminal}$.
Finally, the particular state $\tikzfig{ket-1}$ is plugged on top of the root: $\tikzfig{QMDD2ZH-root}$ (if the root has height $<H$ we technically have empty layers on top of the root, hence the $\tikzfig{ket-plus}$'s).\\[1ex]
The full SQMDD is then mapped as follows, where $L$ is the layer of interest, and $A$ and $B$ will themselves be inductively decomposed layer by layer:
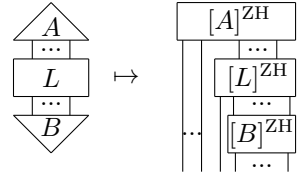
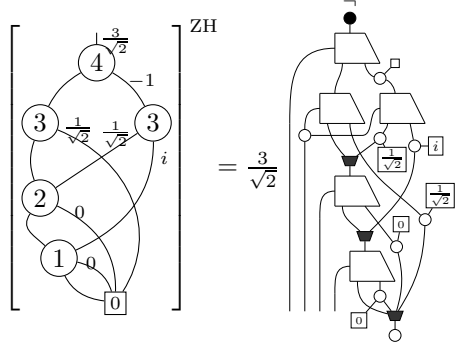
\[\tikzfig{QMDD2ZH-full}\]
\begin{example}~\\
$
\left[\def\fig{example-QMDD}\input{./figures/\fig/\fig_02.tikz}\right]^{\operatorname{ZH}}\!\!=\frac3{\sqrt2}~\def\fig{QMDD2ZH-example}\scalebox{0.8}{\input{./figures/\fig/\fig_02.tikz}}
$
\end{example}

This interpretation $[.]^{\operatorname{ZH}}$ was not chosen at random: it builds a quantum state with the intended semantics.

\begin{proposition}
\label{prop:preserved-semantics}
For any SQMMD $\mathfrak D$, $\interp{[\mathfrak D]^{\operatorname{ZH}}}=\interp{\mathfrak D}$.
\end{proposition}

\begin{proof}
Since $\tikzfig{ket-0}$ and $\tikzfig{ket-1}$ represent $\ket0$ and $\ket1$ respectively, we have, for any ZH-diagram $D:0\to n$:
\def\fig{D-decomp}
\begin{align*}
\interp{\input{./figures/\fig/\fig_00.tikz}}
=\ket0\otimes\interp{\input{./figures/\fig/\fig_01.tikz}}
+\ket1\otimes\interp{\input{./figures/\fig/\fig_02.tikz}}
\end{align*}
Now, let $\mathfrak{D}$ be an SQMDD, and $D:=[\mathfrak{D}]^{\operatorname{ZH}}$. We proceed by induction on $H$ the height of $\mathfrak D$, where the base case is obvious. We then focus on the root $u_0$. 
 We need to distinguish two cases:\\
$\bullet$ $H>h(u_0)$: In this case, $\ell(\mathfrak D)=r(\mathfrak D)$ which entails $\interp{\mathfrak D}=(\ket0+\ket1)\otimes\interp{\ell(\mathfrak D)}$. It can be easily seen that $[\mathfrak D]^{\operatorname{ZH}}=\tikzfig{ket-plus}~[\ell(\mathfrak D)]^{\operatorname{ZH}}$, hence $\interp{[\mathfrak D]^{\operatorname{ZH}}}=\interp{\tikzfig{ket-plus}}\otimes\interp{[\ell(\mathfrak D)]^{\operatorname{ZH}}}=(\ket0+\ket1)\otimes\interp{[\ell(\mathfrak D)]^{\operatorname{ZH}}}$\\
$\bullet$ $H=h(u_0)$: 
Then, \def\fig{root-singled-out}$D =\input{./figures/\fig/\fig_00.tikz}$. We hence have:
\begin{align*}
\interp{[\mathfrak D]^{\operatorname{ZH}}}&=\interp{\input{./figures/\fig/\fig_00.tikz}}
=\ket0\otimes\interp{\input{./figures/\fig/\fig_01.tikz}}+\ket1\otimes\interp{\input{./figures/\fig/\fig_02.tikz}}\\
&=\def\fig{ZH-QMDD-induction}
\ket0\otimes\interp{~\input{./figures/\fig/\fig_00.tikz}~}+\ket1\otimes\interp{~\input{./figures/\fig/\fig_01.tikz}~}
\end{align*}
where the last equality is obtained thanks to Lemmas \ref{lem:ket-0-gadget-ctrl} and \ref{lem:ket-1-gadget-ctrl}.
\def\fig{ZH-QMDD-induction}
We now need to show that $\interp{~\input{./figures/\fig/\fig_00.tikz}~}=\interp{[\ell(\mathfrak D)]^{\operatorname{ZH}}}$ and similarly with the right hand side. We can actually show that we can reduce $\input{./figures/\fig/\fig_00.tikz}$ to $[\ell(\mathfrak D)]^{\operatorname{ZH}}$ using the following rewrites (provable in ZH):
\\
\noindent
\begin{minipage}{0.3\columnwidth}
\begin{equation}
\label{eq:ket1-weight}
\tikzfig{rewrite-ket-1-through-weight}
\end{equation}
\end{minipage}\hfill
\begin{minipage}{0.3\columnwidth}
\begin{equation}
\label{eq:ket0-gadget}
\def\fig{gadget-ket0-top}\input{./figures/\fig/\fig_00.tikz}~\to~\input{./figures/\fig/\fig_04.tikz}
\end{equation}
\end{minipage}\hfill
\begin{minipage}{0.3\columnwidth}
\begin{equation}
\label{eq:ket0-weight}
\tikzfig{rewrite-ket-0-through-weight}
\end{equation}
\end{minipage}\\
\begin{minipage}{0.3\columnwidth}
\begin{equation}
\label{eq:ket0-monoid}
\tikzfig{rewrite-monoid-unit}
\end{equation}
\end{minipage}\hfill
\begin{minipage}{0.3\columnwidth}
\begin{equation}
\label{eq:Z-merge}
\def\fig{rule-Z-spider}\input{./figures/\fig/\fig_00.tikz}~\to~\input{./figures/\fig/\fig_01.tikz}
\end{equation}
\end{minipage}\hfill
\begin{minipage}{0.3\columnwidth}
~
\end{minipage}\\
%
Rewrite \ref{eq:ket1-weight} ensures that if the left wire was weighted, the weight itself gets factored in the overall scalar.\\
Rewrites \ref{eq:ket0-gadget} and \ref{eq:ket0-weight} destroy all the nodes that are not descendent of the left child of the root.\\
Rewrite \ref{eq:ket0-monoid} dictates that if $\tikzfig{ket-0}$ arrives at a node with several parents, the behaviour depends on what happens to the others parents (if all parents are destroyed, $\def\fig{monoid-def}\input{./figures/\fig/\fig_00.tikz}$ will get $\tikzfig{ket-0}$ to all its inputs, which will result in $\tikzfig{ket-0}$, hence pursuing the destruction of subsequent nodes). 
It also shows what happens to $\tikzfig{ket-0}$ when it arrives at the terminal node.\\
Finally, Rewite \ref{eq:Z-merge} simply normalises the connected Z-spiders one can obtain from Rewrite \ref{eq:ket0-gadget}.

This rewrite strategy goes on as long as some $\tikzfig{ket-0}$ exist in the diagram, until they all disappear from Rewrite \ref{eq:ket0-monoid}, in which situation we get the diagram one would have obtained from $\ell(\mathfrak D)$. Similarly, we see that $\def\fig{ZH-QMDD-induction}\input{./figures/\fig/\fig_01.tikz}$ reduces to $[r(\mathfrak D)]^{\operatorname{ZH}}$. Hence, by soundness of the rewrite strategy, we do have:
\begin{align*}
\interp{[\mathfrak D]^{\operatorname{ZH}}} = \ket0\otimes\interp{[\ell(\mathfrak D)]^{\operatorname{ZH}}}+\ket1\otimes\interp{[r(\mathfrak D)]^{\operatorname{ZH}}}
\end{align*}
which by induction hypothesis means $\interp{[\mathfrak D]^{\operatorname{ZH}}}=\interp{\mathfrak D}$.
\end{proof}

This proof introduces a small rewrite strategy, that will be used in the following, in particular to simplify a ZH-diagram in SQMDD form.

\section{Setting a ZH-Diagram in SQMDD Form}
\label{sec:SQMDD-form}

If any SQMDD can be turned into a ZH-diagram, the reciprocal requires some work. In the following, we describe an algorithm that turns any ZH-diagram into SQMDD form, i.e.~into a ZH-diagram that is in the image of $[.]^{\operatorname{ZH}}$, using its equational theory.

\subsection{SQMDD Reduction}
\label{sec:reduction-strategy}

We start to show that all the simplification rules for SQMDDs can be derived directly into the ZH-Calculus. For this, we need the following lemmas:

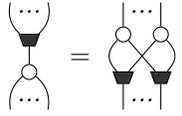
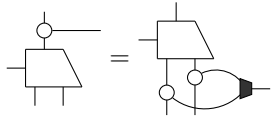
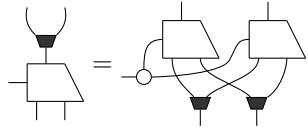
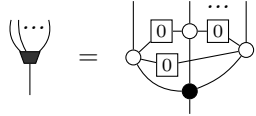
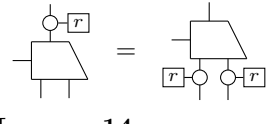
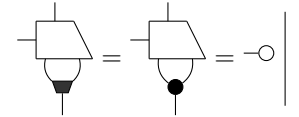
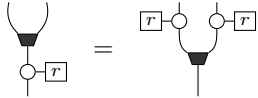
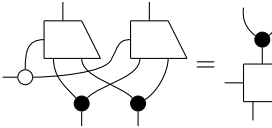
\begin{multicols}{3}

\begin{lemma}
\label{lem:bialgebra-monoid-gn}
$$\tikzfig{gn-monoid-bialgebra-gen}$$
\end{lemma}

\begin{lemma}
\label{lem:n-monoid-decomp}
\def\fig{monoid-gen}
\begin{align*}
\input{./figures/\fig/\fig_04.tikz}
\eq{}\input{./figures/\fig/\fig_03.tikz}
\end{align*}
\end{lemma}

\begin{lemma}
\label{lem:diagonal-distribution-monoid}
\def\fig{diagonal-through-monoid-fin}
\begin{align*}
\input{./figures/\fig/\fig_00.tikz}
\eq{}\input{./figures/\fig/\fig_04.tikz}
\end{align*}
\end{lemma}

\begin{lemma}
\label{lem:gn-through-gadget}
\def\fig{copy-through-gadget}
\begin{align*}
\input{./figures/\fig/\fig_00.tikz}
\eq[]{}\input{./figures/\fig/\fig_04.tikz}
\end{align*}
\end{lemma}

\begin{lemma}
\label{lem:diagonal-distribution-gadget}
\def\fig{diagonal-through-gadget}
\begin{align*}
\input{./figures/\fig/\fig_00.tikz}
\eq{}\input{./figures/\fig/\fig_04.tikz}
\end{align*}
\end{lemma}

\begin{lemma}
\label{lem:bialgebra-rn-gadget}
\def\fig{rn-gadget-bialgebra}
\begin{align*}
\input{./figures/\fig/\fig_00.tikz}
\eq[]{}\input{./figures/\fig/\fig_06.tikz}
\end{align*}
\end{lemma}

\begin{lemma}
\label{lem:bialgebra-monoid-gadget}
\def\fig{monoid-gadget-bialgebra}
\begin{align*}
\input{./figures/\fig/\fig_00.tikz}
\eq[]{}\input{./figures/\fig/\fig_04.tikz}
\end{align*}
\end{lemma}

\begin{lemma}
\label{lem:special-monoid-gadget}
\def\fig{monoid-gadget-special}
\begin{align*}
\input{./figures/\fig/\fig_00.tikz}
\eq[]{}\input{./figures/\fig/\fig_01.tikz}
\eq[]{}\input{./figures/\fig/\fig_04.tikz}
\end{align*}
\end{lemma}

\end{multicols}

\begin{proof}
Proofs for these lemmas start in appendix at page \pageref{prf:bialgebra-monoid-gn}.
\end{proof}

\begin{proposition}
For any simplification rule $\mathfrak D_1 \overset r\to \mathfrak D_2$, the following diagram commutes:
$$\tikzfig{simplification-rules-ZH}$$
\end{proposition}

\begin{proof}
We start with the first rewrite of Figure \ref{fig:QMDD-simplification}, where $a\neq0$:
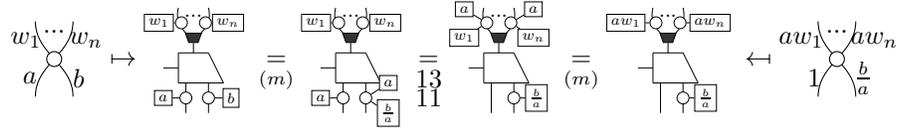
\begin{align*}
\def\fig{QMDD-rule-scalar-distrib}\input{./figures/\fig/\fig_00.tikz}
\mapsto\def\fig{ZH-rule-scalar-distrib}
\scalebox{0.8}{\input{./figures/\fig/\fig_00.tikz}}
&\def\fig{ZH-rule-scalar-distrib}
\eq[]{(m)}\scalebox{0.8}{\input{./figures/\fig/\fig_01.tikz}}
\def\fig{ZH-rule-scalar-distrib}
\eq[]{\ref{lem:diagonal-distribution-gadget}\\\ref{lem:diagonal-distribution-monoid}}\scalebox{0.8}{\input{./figures/\fig/\fig_03.tikz}}
\eq[]{(m)}\scalebox{0.8}{\input{./figures/\fig/\fig_04.tikz}}
\mapsfrom\def\fig{QMDD-rule-scalar-distrib}\input{./figures/\fig/\fig_01.tikz}
\end{align*}
The proof for the rewrite where $a=0$ and $b$ diffuses instead is similar. When applied to the root, the proofs are the same with the additional equality \tikzfig{rewrite-ket-1-through-weight}.


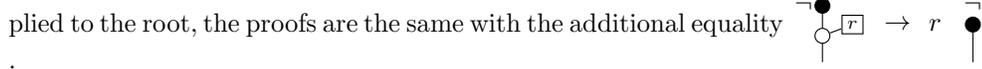
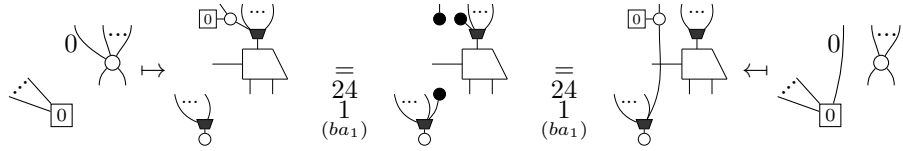
\begin{align*}
\def\fig{QMDD-rule-0-to-terminal}
\input{./figures/\fig/\fig_00.tikz}\mapsto
\def\fig{ZH-rule-0-to-terminal}
\scalebox{0.8}{\input{./figures/\fig/\fig_00.tikz}}
\eq{\ref{lem:0-box}\\\ref{lem:ket-0-monoid}\\(ba_1)}\scalebox{0.8}{\input{./figures/\fig/\fig_01.tikz}}
\eq{\ref{lem:0-box}\\\ref{lem:ket-0-monoid}\\(ba_1)}\scalebox{0.8}{\input{./figures/\fig/\fig_02.tikz}}
\mapsfrom
\def\fig{QMDD-rule-0-to-terminal}
\input{./figures/\fig/\fig_01.tikz}
\end{align*}

The rule \tikzfig{QMDD-rule-remove-coleaves} is a direct consequence of the rewrite strategy of the proof of Proposition \ref{prop:preserved-semantics}.

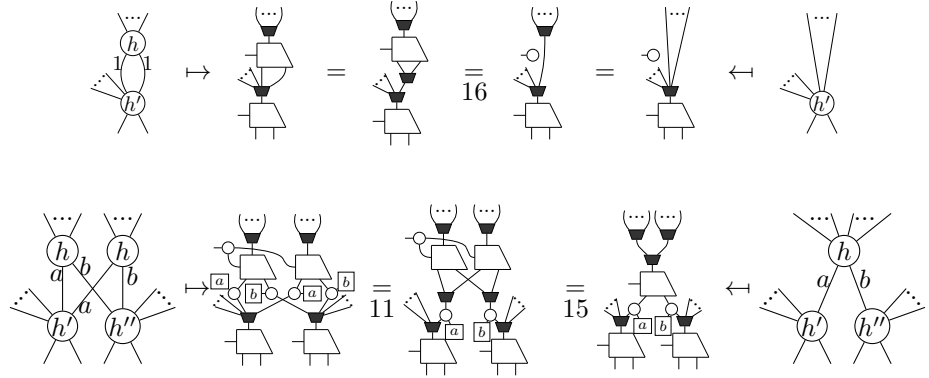
\begin{align*}
\def\fig{QMDD-rule-variable-skipping}\scalebox{0.8}{\input{./figures/\fig/\fig_00.tikz}}
~~\mapsto~~\def\fig{ZH-rule-variable-skipping}
\scalebox{0.8}{\input{./figures/\fig/\fig_00.tikz}}
\eq{}\scalebox{0.8}{\input{./figures/\fig/\fig_01.tikz}}
\eq{\ref{lem:special-monoid-gadget}}\scalebox{0.8}{\input{./figures/\fig/\fig_02.tikz}}
\eq{}\scalebox{0.8}{\input{./figures/\fig/\fig_03.tikz}}
\def\fig{QMDD-rule-variable-skipping}~~\mapsfrom~~\scalebox{0.8}{\input{./figures/\fig/\fig_01.tikz}}
\end{align*}
%
\begin{align*}
\def\fig{QMDD-rule-regrouping-general}\input{./figures/\fig/\fig_00.tikz}
\mapsto&\def\fig{ZH-rule-regrouping-general}
\scalebox{0.8}{\input{./figures/\fig/\fig_00.tikz}}
\eq[]{\ref{lem:diagonal-distribution-monoid}}\scalebox{0.8}{\input{./figures/\fig/\fig_01.tikz}}
&\def\fig{ZH-rule-regrouping-general}
\eq[]{\ref{lem:bialgebra-monoid-gadget}}\scalebox{0.8}{\input{./figures/\fig/\fig_02.tikz}}
\mapsfrom\def\fig{QMDD-rule-regrouping-general}\input{./figures/\fig/\fig_01.tikz}
\end{align*}

The case where the two children are the same is very similar:
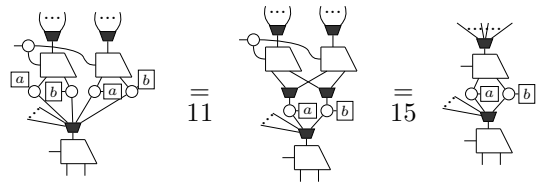
\def\fig{ZH-rule-regrouping-unique-child}
\begin{align*}
\scalebox{0.8}{\input{./figures/\fig/\fig_00.tikz}}
\eq{\ref{lem:diagonal-distribution-monoid}}\scalebox{0.8}{\input{./figures/\fig/\fig_01.tikz}}
\eq{\ref{lem:bialgebra-monoid-gadget}}\scalebox{0.8}{\input{./figures/\fig/\fig_02.tikz}}
\end{align*}
\end{proof}
Thanks to this proof, we now have a strategy to systematically reduce any ZH-diagram in SQMDD form. We then show how to put a ZH-diagram in SQMDD form in the first place.

\subsection{SQMDD form of Generators and Compositions}
\label{sec:SQMDD-gen-comp}

We now show that all the generators of the ZH-Calculus can be put in SQMDD form, and furthermore that the compositions of diagrams in SQMDD form can be put in SQMDD form.

\begin{proposition}
\label{prop:pseudo-SQMDD}
Suppose $D'$ differs from a diagram $D$ in SQMDD form by replacing some of its \def\fig{monoid-def}$\input{./figures/\fig/\fig_00.tikz}$ by $\tikzfig{rn-2-1}$. We then have: $\operatorname{ZH}\vdash D = D'$.
\end{proposition}

\begin{proof}
In appendix at page \pageref{prf:pseudo-SQMDD}.
\end{proof}

This result actually still holds for any commutative monoid whose neutral element is $\tikzfig{ket-0}$. A straightforward analysis yields that such monoids are of the form $\begin{pmatrix}1&0&0&a\\0&1&1&b\end{pmatrix}$ for $a,b\in\mathbb C$.

This result will be used in the proposition that follows, but notice that it can also be used to simplify the ZH-diagram obtained from an SQMDD.

Now, we can show that all the generators can be set in SQMDD form:
\begin{proposition}
\label{prop:generators-NF}
The generators of the ZH-Calculus can be set in SQMDD form:

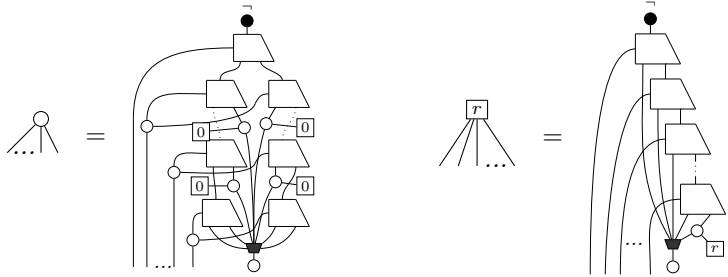
\begin{minipage}{0.45\columnwidth}
\def\fig{Z-spider-NF-3}
\begin{align*}
\input{./figures/\fig/\fig_00.tikz}
\eq{}\scalebox{0.8}{\input{./figures/\fig/\fig_06.tikz}}
\end{align*}
\end{minipage}
\begin{minipage}{0.45\columnwidth}
\def\fig{H-spider-NF}
\begin{align*}
\input{./figures/\fig/\fig_00.tikz}
\eq{}\scalebox{0.8}{\input{./figures/\fig/\fig_04.tikz}}
\end{align*}
\end{minipage}
\end{proposition}

\begin{proof}
In appendix at page \pageref{prf:generators-NF}.
\end{proof}

We then need to show that the composition of diagrams in SQMDD form can be put in SQMDD form.

\begin{proposition}
\label{prop:SQMDD-tensor}
The parallel composition of ZH-diagrams in SQMDD form can be put in SQMDD form, by joining the root of the right diagram to the terminal node of the diagram on the left.
\end{proposition}

\begin{proof}
In appendix at page \pageref{prf:SQMDD-tensor}.
\end{proof}

Since we work only with quantum states after using the map/state duality, what accounts for sequential composition is the linking of two ZH-diagrams in SQMDD form through $\tikzfig{cup}$. This element can be decomposed into two primitives: $\tikzfig{cup}\eq{}\tikzfig{cup-from-Z}$. We hence show that applying each of these two primitives to a ZH-diagram in SQMDD form can be put in SQMDD form. To simplify the proof, we first show that swapping two outputs of a state in SQMDD can be put in SQMDD form.

\begin{proposition}
\label{prop:SQMDD-swap}
Let $D$ be a ZH-diagram in SQMDD form. Applying $\tikzfig{swap}$ on two of its outputs can be put in SQMDD form. By repeated application of $\tikzfig{swap}$, any permutation of the outputs of $D$ can be put in SQMDD form.
\end{proposition}

\begin{proof}
In appendix at page \pageref{prf:SQMDD-swap}.
\end{proof}

\begin{proposition}
\label{prop:SQMDD-gn-2-1}
Let $D$ be a ZH-diagram in SQMDD form. Applying $\tikzfig{gn-2-1}$ on two of its outputs can be put in SQMDD form.
\end{proposition}

\begin{proof}
In appendix at page \pageref{prf:SQMDD-gn-2-1}.
\end{proof}

\begin{proposition}
\label{prop:SQMDD-gn-1-0}
Let $D$ be a ZH-diagram in SQMDD form. Applying $\tikzfig{bra-plus}$ on any of its outputs can be put in SQMDD form.
\end{proposition}

\begin{proof}
In appendix at page \pageref{prf:SQMDD-gn-1-0}.
\end{proof}

Combining the previous propositions, we get the expected result:
\begin{theorem}
Any ZH-diagram can be put in reduced SQMDD form.
\end{theorem}

Notice that this theorem gives an algorithm to translate any ZH-diagram into an SQMDD. Indeed, since a diagram $D$ in reduced SQMDD form is in the image of $[.]^{\operatorname{ZH}}$ (which is injective), it suffices to take its inverse.

%
%
%
%

\bibliography{../quantum-ref}

\appendix

\section{Useful Lemmas for the ZH-Calculus}
\label{sec:useful-lemmas}

\begin{multicols}{3}

\begin{lemma}
\label{lem:Z-H-multiple-links}
$$\tikzfig{Z-H-multiple-links}$$
\end{lemma}

\begin{lemma}
\label{lem:Not-through-H-box}
$$\tikzfig{Not-through-H-spider}$$
\end{lemma}

\begin{lemma}
\label{lem:bialgebra-Z-AND}
\def\fig{Z-And-bialgebra}
$$\frac12~\input{./figures/\fig/\fig_00.tikz}
\eq{}\frac1{2^n}~\input{./figures/\fig/\fig_01.tikz}$$
\end{lemma}

\begin{lemma}
\label{lem:X-spider}
$$\tikzfig{X-spider-1}$$
and
$$\tikzfig{X-spider-2}$$
\end{lemma}

\begin{lemma}
\label{lem:classical-on-H-box}
$$\tikzfig{ket-0-on-H-spider}$$
and
$$\tikzfig{ket-1-on-H-spider}$$
\end{lemma}

\begin{lemma}
\label{lem:product-gen}
$$\tikzfig{rule-product-gen}$$
\end{lemma}

\begin{lemma}
\label{lem:CNOT-on-H-legs}
\def\fig{CNot-through-H-spider}
\begin{align*}
\input{./figures/\fig/\fig_00.tikz}
\eq{}\input{./figures/\fig/\fig_05.tikz}
\end{align*}
\end{lemma}

\begin{lemma}
\label{lem:0-box}
\def\fig{H-spider-0}
\begin{align*}
\input{./figures/\fig/\fig_00.tikz}
\eq{}\frac12~\input{./figures/\fig/\fig_05.tikz}
\end{align*}
\end{lemma}

\begin{lemma}
\label{lem:1-box}
\def\fig{H-spider-1}
\begin{align*}
\input{./figures/\fig/\fig_00.tikz}
\eq{}\input{./figures/\fig/\fig_04.tikz}
\end{align*}
\end{lemma}

\begin{lemma}
\label{lem:bialgebra-neg}
$$\tikzfig{bialgebra-neg}$$
\end{lemma}

\begin{lemma}
\label{lem:2-box}
$$\tikzfig{box-2}$$
\end{lemma}

\begin{lemma}
\label{lem:gn-0-0}
$$\tikzfig{gn-0-0}\eq{}2$$
\end{lemma}

\begin{lemma}
\label{lem:hopf-gn-H}
$$\tikzfig{hopf-gn-H-box}$$
\end{lemma}

\begin{lemma}
\label{lem:hopf}
$$\tikzfig{hopf}$$
\end{lemma}

\end{multicols}

\section{Proofs of Section \ref{sec:ZH}}

\begin{proof}[Proof of Lemmas \ref{lem:ket-0-monoid} and \ref{lem:ket-1-monoid}]
\phantomsection\label{prf:monoid-ket-0-1}
\def\fig{monoid-ket0}
\begin{align*}
\input{./figures/\fig/\fig_00.tikz}
\eq{}\input{./figures/\fig/\fig_01.tikz}
\eq{(ba_1)}\input{./figures/\fig/\fig_02.tikz}
\eq{\ref{lem:classical-on-H-box}}\input{./figures/\fig/\fig_03.tikz}
\eq{(hs)\\(id)}\input{./figures/\fig/\fig_04.tikz}
\end{align*}

\def\fig{monoid-ket1}
\begin{align*}
\input{./figures/\fig/\fig_00.tikz}
\eq{}\frac12~\input{./figures/\fig/\fig_01.tikz}
\eq{\ref{lem:bialgebra-neg}\\\ref{lem:X-spider}}\frac12~\input{./figures/\fig/\fig_02.tikz}
\eq{\ref{lem:classical-on-H-box}}\input{./figures/\fig/\fig_03.tikz}
\eq{(ba_1)\\\ref{lem:X-spider}}\input{./figures/\fig/\fig_04.tikz}
\end{align*}

\end{proof}

\begin{proof}[Proof of Lemmas \ref{lem:ket-0-gadget-ctrl}, \ref{lem:swapped-gadget-legs}, \ref{lem:ket-1-gadget-ctrl} and \ref{lem:ket-0-top-gadget}]
\phantomsection\label{prf:gadget-ket-0-1}
\def\fig{gadget-0-ctrld}
\begin{align*}
\input{./figures/\fig/\fig_00.tikz}
\eq{}\frac12~\input{./figures/\fig/\fig_01.tikz}
\eq{(ba_2)}\frac12~\input{./figures/\fig/\fig_02.tikz}
\eq{(zs)\\(id)}\input{./figures/\fig/\fig_03.tikz}
\eq{(ba_1)\\\ref{lem:X-spider}\\(hh)}\input{./figures/\fig/\fig_04.tikz}
\end{align*}

\def\fig{gadget-swap}
\begin{align*}
\input{./figures/\fig/\fig_00.tikz}
\eq[]{}\frac12~\input{./figures/\fig/\fig_01.tikz}
\eq[]{(ba_1)}\frac12~\input{./figures/\fig/\fig_02.tikz}
\eq[]{(hh)}\frac12~\input{./figures/\fig/\fig_03.tikz}
\eq[]{\ref{lem:Not-through-H-box}}\frac12~\input{./figures/\fig/\fig_04.tikz}
\eq[]{}\input{./figures/\fig/\fig_05.tikz}
\end{align*}

\def\fig{gadget-1-ctrld}
\begin{align*}
\input{./figures/\fig/\fig_00.tikz}
\eq{\ref{lem:X-spider}\\\ref{lem:swapped-gadget-legs}}\input{./figures/\fig/\fig_01.tikz}
\eq{\ref{lem:ket-0-gadget-ctrl}}\input{./figures/\fig/\fig_02.tikz}
\eq{}\input{./figures/\fig/\fig_03.tikz}
\end{align*}

\def\fig{gadget-ket0-top}
\begin{align*}
\input{./figures/\fig/\fig_00.tikz}
\eq{}\frac12~\input{./figures/\fig/\fig_01.tikz}
\eq{(ba_1)}\frac12~\input{./figures/\fig/\fig_02.tikz}
\eq{\ref{lem:classical-on-H-box}}\frac12~\input{./figures/\fig/\fig_03.tikz}
\eq{(ba_1)}\input{./figures/\fig/\fig_04.tikz}
\end{align*}

\end{proof}

\section{Proofs of Section \ref{sec:reduction-strategy}}

We now aim to show Lemmas \ref{lem:bialgebra-monoid-gn}, \ref{lem:n-monoid-decomp}, \ref{lem:diagonal-distribution-monoid}, \ref{lem:gn-through-gadget}, \ref{lem:diagonal-distribution-gadget}, \ref{lem:bialgebra-rn-gadget}, \ref{lem:bialgebra-monoid-gadget} and \ref{lem:special-monoid-gadget}. We will show them one by one, after introducing and proving auxiliary lemmas in between.

\begin{multicols}{3}
\begin{lemma}
\label{lem:parallel-0-box}
\def\fig{triangle-projector}
\begin{align*}
\input{./figures/\fig/\fig_00.tikz}
\eq{}\input{./figures/\fig/\fig_04.tikz}
\end{align*}
\end{lemma}
\begin{lemma}
\label{lem:projector-absorbed-by-monoid}
\def\fig{triangle-absorbed-by-monoid}
\begin{align*}
\input{./figures/\fig/\fig_00.tikz}
\eq{}\input{./figures/\fig/\fig_03.tikz}
\end{align*}
\end{lemma}
\end{multicols}

\begin{proof}
\def\fig{triangle-projector}
\begin{align*}
\input{./figures/\fig/\fig_00.tikz}
\eq{\ref{lem:0-box}}\frac14~\input{./figures/\fig/\fig_01.tikz}
\eq{(ba_2)}\frac14~\input{./figures/\fig/\fig_02.tikz}
\eq{(ba_1)\\\ref{lem:gn-0-0}}\frac12~\input{./figures/\fig/\fig_03.tikz}
\eq{\ref{lem:0-box}}\input{./figures/\fig/\fig_04.tikz}
\end{align*}

\def\fig{triangle-absorbed-by-monoid}
\begin{align*}
\input{./figures/\fig/\fig_00.tikz}
\eq{}\input{./figures/\fig/\fig_01.tikz}
\eq{\ref{lem:parallel-0-box}}\input{./figures/\fig/\fig_02.tikz}
\eq{}\input{./figures/\fig/\fig_03.tikz}
\end{align*}
\end{proof}

\begin{proof}[Proof of Lemma \ref{lem:bialgebra-monoid-gn}]
\phantomsection\label{prf:bialgebra-monoid-gn}
\def\fig{gn-monoid-bialgebra}
\begin{align*}
\input{./figures/\fig/\fig_00.tikz}
\eq{}\input{./figures/\fig/\fig_01.tikz}
\eq{(ba_1)}\input{./figures/\fig/\fig_02.tikz}
\eq{\ref{lem:parallel-0-box}\\(zs)}\input{./figures/\fig/\fig_03.tikz}
\eq{}\input{./figures/\fig/\fig_04.tikz}
\end{align*}
\end{proof}

\begin{multicols}{3}

\begin{lemma}
\label{lem:0-box-under-monoid}
\def\fig{monoid-with-triangle-2}
\begin{align*}
\input{./figures/\fig/\fig_00.tikz}
\eq{}\input{./figures/\fig/\fig_09.tikz}
\end{align*}
\end{lemma}
\begin{lemma}
\label{lem:projector-distribution-monoid}
$$\tikzfig{triangle-distribution-through-monoid-gen}$$
\end{lemma}
\end{multicols}

\begin{proof}
First, we have:
\def\fig{monoid-with-triangle-1}
\begin{align*}
\input{./figures/\fig/\fig_00.tikz}
\eq{\ref{lem:0-box}}\frac12~\input{./figures/\fig/\fig_01.tikz}
\eq{}\frac14~\input{./figures/\fig/\fig_02.tikz}
\eq{(ba_1)}\frac14~\input{./figures/\fig/\fig_03.tikz}
\eq{\ref{lem:CNOT-on-H-legs}}\frac14~\input{./figures/\fig/\fig_04.tikz}\\
\eq{(o)}\frac12~\input{./figures/\fig/\fig_05.tikz}
\eq{\ref{lem:CNOT-on-H-legs}}\frac12~\input{./figures/\fig/\fig_06.tikz}
\eq{(ba_1)}\frac12~\input{./figures/\fig/\fig_07.tikz}
\eq{(ba_2)}\frac12~\input{./figures/\fig/\fig_08.tikz}
\end{align*}
Then:
\def\fig{monoid-with-triangle-2}
\begin{align*}
\input{./figures/\fig/\fig_00.tikz}
\eq{}\input{./figures/\fig/\fig_01.tikz}
\eq{\ref{lem:projector-absorbed-by-monoid}}\input{./figures/\fig/\fig_02.tikz}
\eq{\ref{lem:0-box}}\frac12~\input{./figures/\fig/\fig_03.tikz}
\eq{}\frac12~\input{./figures/\fig/\fig_04.tikz}
\eq{(ba_2)}\frac12~\input{./figures/\fig/\fig_05.tikz}\\
\eq{(zs)\\(id)}\frac14~\input{./figures/\fig/\fig_06.tikz}
\eq{(ba_2)}\frac14~\input{./figures/\fig/\fig_07.tikz}
\eq{(ba_1)}\frac14~\input{./figures/\fig/\fig_08.tikz}
\eq{\ref{lem:0-box}}\input{./figures/\fig/\fig_09.tikz}
\end{align*}

\def\fig{triangle-distribution-through-monoid}
\begin{align*}
\input{./figures/\fig/\fig_00.tikz}
\eq{\ref{lem:bialgebra-monoid-gn}}\input{./figures/\fig/\fig_01.tikz}
\eq{\ref{lem:0-box-under-monoid}}\input{./figures/\fig/\fig_02.tikz}
\eq{\ref{lem:projector-absorbed-by-monoid}}\input{./figures/\fig/\fig_03.tikz}
\end{align*}
\end{proof}

\begin{proof}[Proof of Lemmas \ref{lem:n-monoid-decomp} and \ref{lem:diagonal-distribution-monoid}]

\def\fig{monoid-gen}
\begin{align*}
\input{./figures/\fig/\fig_04.tikz}
\eq{}\input{./figures/\fig/\fig_00.tikz}
\eq{}\input{./figures/\fig/\fig_01.tikz}
\eq{\ref{lem:projector-distribution-monoid}}\input{./figures/\fig/\fig_02.tikz}
\eq{ind.}\input{./figures/\fig/\fig_03.tikz}
\end{align*}

First, we have:
\def\fig{diagonal-through-monoid}
\begin{align*}
\input{./figures/\fig/\fig_00.tikz}
\eq{(hh)}\frac14\input{./figures/\fig/\fig_01.tikz}
\eq{(ba_1)}\frac14\input{./figures/\fig/\fig_02.tikz}
\eq{(ba_2)}\frac14\input{./figures/\fig/\fig_03.tikz}
\eq{\ref{lem:Not-through-H-box}}\frac14\input{./figures/\fig/\fig_04.tikz}\\
\eq{(hs)}\input{./figures/\fig/\fig_05.tikz}
\eq{\ref{lem:CNOT-on-H-legs}}\input{./figures/\fig/\fig_06.tikz}
\eq{\ref{lem:X-spider}\\(ba_1)\\\ref{lem:bialgebra-neg}}\input{./figures/\fig/\fig_07.tikz}
\eq{(i)}\input{./figures/\fig/\fig_08.tikz}
\eq{(zs)}\input{./figures/\fig/\fig_09.tikz}
\end{align*}
Then:
\def\fig{diagonal-through-monoid-fin}
\begin{align*}
\input{./figures/\fig/\fig_00.tikz}
\eq{}\frac12\input{./figures/\fig/\fig_01.tikz}
\eq{(ba_1)}\frac12\input{./figures/\fig/\fig_02.tikz}
\eq{}\frac12\input{./figures/\fig/\fig_03.tikz}
\eq{}\input{./figures/\fig/\fig_04.tikz}
\end{align*}
\end{proof}

\begin{multicols}{3}

\begin{lemma}
\label{lem:projector-on-gadget-legs}
\def\fig{triangle-absorbed-by-gadget}
\begin{align*}
\input{./figures/\fig/\fig_00.tikz}
\eq{}\input{./figures/\fig/\fig_09.tikz}
\end{align*}
\end{lemma}

\end{multicols}

\begin{proof}
\def\fig{triangle-absorbed-by-gadget}
\begin{align*}
\input{./figures/\fig/\fig_00.tikz}
&\eq{\ref{lem:0-box}}\frac12\input{./figures/\fig/\fig_01.tikz}
\eq{}\frac14\input{./figures/\fig/\fig_02.tikz}
\eq{(ba_1)}\frac14\input{./figures/\fig/\fig_03.tikz}\\
&\eq{\ref{lem:CNOT-on-H-legs}}\frac14\input{./figures/\fig/\fig_04.tikz}
\eq{(hh)}\frac14\input{./figures/\fig/\fig_05.tikz}
\eq{(ba_2)}\frac14\input{./figures/\fig/\fig_06.tikz}\\
&\eq{\ref{lem:CNOT-on-H-legs}}\frac14\input{./figures/\fig/\fig_07.tikz}
\eq{\ref{lem:classical-on-H-box}\\(zs)\\\ref{lem:gn-0-0}}\frac12\input{./figures/\fig/\fig_08.tikz}
\eq{}\input{./figures/\fig/\fig_09.tikz}
\end{align*}
\end{proof}

\begin{proof}[Proof of Lemmas \ref{lem:gn-through-gadget} and \ref{lem:diagonal-distribution-gadget}]

\def\fig{copy-through-gadget}
\begin{align*}
\input{./figures/\fig/\fig_00.tikz}
\eq{(zs)}\frac12\input{./figures/\fig/\fig_01.tikz}
\eq{(ba_1)}\frac12\input{./figures/\fig/\fig_02.tikz}
\eq{\ref{lem:projector-on-gadget-legs}}\frac14\input{./figures/\fig/\fig_03.tikz}
\eq{}\input{./figures/\fig/\fig_04.tikz}
\end{align*}

\def\fig{diagonal-through-gadget}
\begin{align*}
\input{./figures/\fig/\fig_00.tikz}
\eq{\ref{lem:gn-through-gadget}}\input{./figures/\fig/\fig_01.tikz}
\eq{\ref{lem:diagonal-distribution-monoid}}\input{./figures/\fig/\fig_02.tikz}
\eq{(ba_1)}\frac12\input{./figures/\fig/\fig_03.tikz}
\eq{\ref{lem:projector-on-gadget-legs}}\input{./figures/\fig/\fig_04.tikz}
\end{align*}

\end{proof}

\begin{multicols}{3}
\begin{lemma}
\label{lem:projector-through-one-gadget}
\def\fig{triangle-through-gadget-3}
\begin{align*}
\input{./figures/\fig/\fig_00.tikz}
\eq[]{}\input{./figures/\fig/\fig_03.tikz}
\end{align*}
\end{lemma}
\begin{lemma}
\label{lem:projector-through-gadget}
\def\fig{triangle-through-gadget-2}
\begin{align*}
\fit{$\input{./figures/\fig/\fig_05.tikz}
=\input{./figures/\fig/\fig_00.tikz}$}
\end{align*}
\end{lemma}
\end{multicols}

\begin{proof}

\def\fig{triangle-through-gadget-3}
\begin{align*}
\input{./figures/\fig/\fig_00.tikz}
\eq{\ref{lem:gn-through-gadget}}\input{./figures/\fig/\fig_01.tikz}
\eq{\ref{lem:0-box-under-monoid}}\input{./figures/\fig/\fig_02.tikz}
\eq{\ref{lem:projector-on-gadget-legs}}\input{./figures/\fig/\fig_03.tikz}
\end{align*}

For Lemma \ref{lem:projector-through-gadget}, first:
\def\fig{triangle-through-gadget-1}
\begin{align*}
\input{./figures/\fig/\fig_00.tikz}
\eq{}\frac1{16}\input{./figures/\fig/\fig_01.tikz}
\eq{(ba_1)}\frac14\input{./figures/\fig/\fig_02.tikz}
\eq{(ba_2)\\(hs)}\input{./figures/\fig/\fig_03.tikz}\\
\eq{\ref{lem:Z-H-multiple-links}\\\ref{lem:Not-through-H-box}}\input{./figures/\fig/\fig_04.tikz}
\eq{\ref{lem:bialgebra-Z-AND}\\(hs)}\input{./figures/\fig/\fig_05.tikz}
\eq{\ref{lem:bialgebra-neg}\\\ref{lem:CNOT-on-H-legs}}\input{./figures/\fig/\fig_06.tikz}\\
\eq{\ref{lem:bialgebra-neg}\\\ref{lem:Z-H-multiple-links}}\input{./figures/\fig/\fig_07.tikz}
\eq{(ba_2)}\frac12\input{./figures/\fig/\fig_08.tikz}
\eq{\ref{lem:Not-through-H-box}}\frac12\input{./figures/\fig/\fig_09.tikz}\\
\eq{(ba_2)}\frac12\input{./figures/\fig/\fig_10.tikz}
\eq{(ba_1)}\input{./figures/\fig/\fig_11.tikz}
\eq{(ba_1)\\\ref{lem:classical-on-H-box}}2\input{./figures/\fig/\fig_12.tikz}
\end{align*}
Then:
\def\fig{triangle-through-gadget-2}
\begin{align*}
\input{./figures/\fig/\fig_00.tikz}
\eq{\ref{lem:0-box}}\frac1{16}\input{./figures/\fig/\fig_01.tikz}
\eq{(ba_1)\\(zs)}\frac1{16}\input{./figures/\fig/\fig_02.tikz}\\
\eq{}\frac18\input{./figures/\fig/\fig_03.tikz}
\eq{(zs)}\frac18\input{./figures/\fig/\fig_04.tikz}
\eq{\ref{lem:0-box}}\input{./figures/\fig/\fig_05.tikz}
\end{align*}

\end{proof}

\begin{proof}[Proof of Lemmas \ref{lem:bialgebra-rn-gadget}, \ref{lem:bialgebra-monoid-gadget}, \ref{lem:special-monoid-gadget}]

\def\fig{rn-gadget-bialgebra}
\begin{align*}
\input{./figures/\fig/\fig_00.tikz}
\eq{}\frac14\input{./figures/\fig/\fig_01.tikz}
\eq{}\frac14\input{./figures/\fig/\fig_02.tikz}
\eq{\ref{lem:X-spider}\\(ba_1)}\frac12\input{./figures/\fig/\fig_03.tikz}\\
\eq{(ba_2)}\frac12\input{./figures/\fig/\fig_04.tikz}
\eq{(ba_1)}\frac12\input{./figures/\fig/\fig_05.tikz}
\eq{}\input{./figures/\fig/\fig_06.tikz}
\end{align*}

\def\fig{monoid-gadget-bialgebra}
\begin{align*}
\input{./figures/\fig/\fig_00.tikz}
\eq{}\input{./figures/\fig/\fig_01.tikz}
\eq{\ref{lem:bialgebra-monoid-gadget}}\input{./figures/\fig/\fig_02.tikz}
\eq{\ref{lem:projector-distribution-monoid}}\input{./figures/\fig/\fig_03.tikz}
\eq{}\input{./figures/\fig/\fig_04.tikz}
\end{align*}

\def\fig{monoid-gadget-special}
\begin{align*}
\input{./figures/\fig/\fig_00.tikz}
\eq{\ref{lem:projector-on-gadget-legs}}\input{./figures/\fig/\fig_01.tikz}
\eq{}\frac12\input{./figures/\fig/\fig_02.tikz}
\eq{\ref{lem:hopf}\\(id)}\frac12\input{./figures/\fig/\fig_03.tikz}
\eq{\ref{lem:bialgebra-Z-AND}\\(zs)\\(id)}\input{./figures/\fig/\fig_04.tikz}
\end{align*}

\end{proof}

\section{Proofs of Section \ref{sec:SQMDD-gen-comp}}

\begin{proof}[Proof of Proposition \ref{prop:pseudo-SQMDD}]
\phantomsection\label{prf:pseudo-SQMDD}
We will start at the bottom of the diagram and work our way up, layer by layer. We first create the gadget \tikzfig{projector} at the leaf:
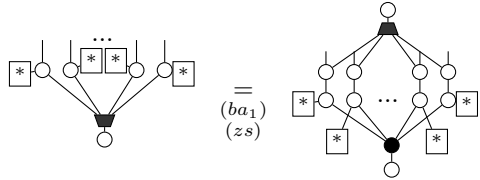
\def\fig{pseudo-SQMDD-leaf}
\begin{align*}
\input{./figures/\fig/\fig_00.tikz}
\eq{(ba_1)\\(zs)}\input{./figures/\fig/\fig_01.tikz}
\end{align*}
and then we move it up. It is easy, using Lemmas \ref{lem:n-monoid-decomp} and \ref{lem:parallel-0-box}, to see that \tikzfig{pseudo-SQMDD-lemma}. Hence, if we consider a particular layer, we get:
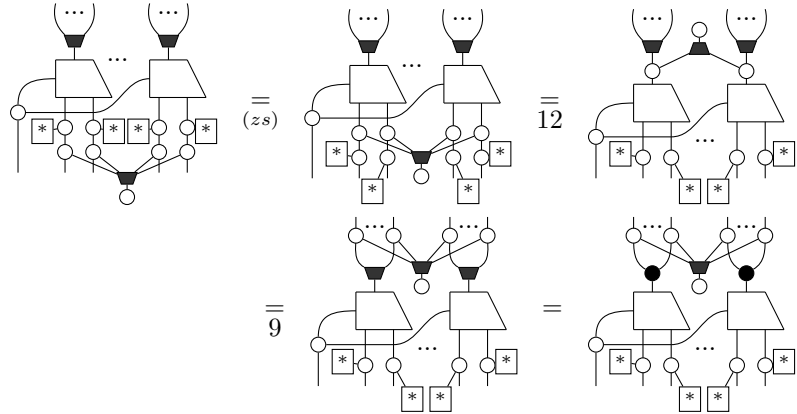
\def\fig{pseudo-SQMDD-layer}
\begin{align*}
\input{./figures/\fig/\fig_00.tikz}
\eq{(zs)}\input{./figures/\fig/\fig_01.tikz}
\eq{\ref{lem:gn-through-gadget}}\input{./figures/\fig/\fig_02.tikz}\\
\eq{\ref{lem:bialgebra-monoid-gn}}\input{./figures/\fig/\fig_03.tikz}
\eq{}\input{./figures/\fig/\fig_04.tikz}
\end{align*}
Notice how we pushed the gadget through the whole layer while replacing some \def\fig{monoid-def}$\input{./figures/\fig/\fig_00.tikz}$ (the ones we want) by $\tikzfig{rn-2-1}$. The gadget will hence move up until it gets to the root, where it can simply be removed.
\end{proof}

\begin{proof}[Proof of Proposition \ref{prop:generators-NF}]
\phantomsection\label{prf:generators-NF}
First, we notice that:
\def\fig{gadget-ketplus-left}
\begin{align}
\label{eq:gadget-ket-plus-left}
\input{./figures/\fig/\fig_00.tikz}
\eq{}\frac12~\input{./figures/\fig/\fig_01.tikz}
\eq{(ba_1)\\(zs)}\frac12~\input{./figures/\fig/\fig_02.tikz}
\eq{(id)}\frac12~\input{./figures/\fig/\fig_03.tikz}
\end{align}
from which can derive the SQMDD form of the H-spider:
\def\fig{H-spider-NF}
\begin{align*}
\input{./figures/\fig/\fig_00.tikz}
\eq[]{(hs)\\\ref{lem:classical-on-H-box}}\frac1{2^n}\input{./figures/\fig/\fig_01.tikz}
\eq[]{(\ref{eq:gadget-ket-plus-left})}\input{./figures/\fig/\fig_02.tikz}
\eq[]{(id)\\(zs)\\(ba_1)}\input{./figures/\fig/\fig_03.tikz}
\eq[]{\ref{prop:pseudo-SQMDD}}\input{./figures/\fig/\fig_04.tikz}
\end{align*}

For the normal form of the Z-spider, we need a few intermediary derivations:
\def\fig{Z-spider-NF-1}
\begin{align}
\label{eq:Z-spider-NF-1}
\input{./figures/\fig/\fig_00.tikz}
\eq{\ref{lem:0-box}}\frac14\input{./figures/\fig/\fig_01.tikz}
\eq{\ref{lem:Not-through-H-box}\\(id)}\frac14\input{./figures/\fig/\fig_02.tikz}
\eq{(ba_2)}\frac14\input{./figures/\fig/\fig_03.tikz}
\eq{(zs)}\frac12\input{./figures/\fig/\fig_04.tikz}
\eq[]{\ref{lem:Z-H-multiple-links}}\frac12\input{./figures/\fig/\fig_05.tikz}
\eq{(hh)\\(zs)\\(id)}\input{./figures/\fig/\fig_06.tikz}
\end{align}

{
\def\fig{Z-spider-NF-2}
\begin{align}
\label{eq:Z-spider-NF-2}
\input{./figures/\fig/\fig_00.tikz}
\eq[]{\ref{lem:swapped-gadget-legs}}\frac14\input{./figures/\fig/\fig_01.tikz}
\eq[]{(ba_1)\\\ref{lem:classical-on-H-box}}\frac14\input{./figures/\fig/\fig_02.tikz}
\eq[]{\ref{lem:bialgebra-neg}\\(zs)}\frac14\input{./figures/\fig/\fig_03.tikz}
\eq[]{(\ref{eq:Z-spider-NF-1})}\input{./figures/\fig/\fig_04.tikz}
\end{align}}

\def\fig{gadget-ket1-top}
\begin{align}
\label{eq:Z-spider-NF-2.5}
\input{./figures/\fig/\fig_00.tikz}
\eq{}\frac12\input{./figures/\fig/\fig_01.tikz}
\eq{\ref{lem:bialgebra-neg}}\frac12\input{./figures/\fig/\fig_02.tikz}
\eq{\ref{lem:classical-on-H-box}\\\ref{lem:X-spider}}\frac12\input{./figures/\fig/\fig_03.tikz}
\eq{(hh)}\input{./figures/\fig/\fig_04.tikz}
\end{align}
Now, we can show how to turn a Z-spider in SQMDD form.
\def\fig{Z-spider-NF-3}
\begin{align*}
\input{./figures/\fig/\fig_00.tikz}
&\eq{(zs)\\\ref{lem:bialgebra-neg}}\input{./figures/\fig/\fig_01.tikz}
\eq{(\ref{eq:Z-spider-NF-2})}\input{./figures/\fig/\fig_02.tikz}
\eq{(\ref{eq:Z-spider-NF-2})}...
\eq{}\input{./figures/\fig/\fig_03.tikz}\\
&\eq{(\ref{eq:Z-spider-NF-2.5})}\input{./figures/\fig/\fig_04.tikz}
\eq{\ref{lem:0-box}\\(ba_1)}\input{./figures/\fig/\fig_05.tikz}
\eq{\ref{prop:pseudo-SQMDD}}\input{./figures/\fig/\fig_06.tikz}
\end{align*}
\end{proof}

\begin{proof}[Proof of Proposition \ref{prop:SQMDD-tensor}]
\phantomsection\label{prf:SQMDD-tensor}
Let $D_1$ and $D_2$ be two ZH-diagrams in SQMDD form. If we define $D_i'$ by making the top and bottom nodes stand out, we can show that:
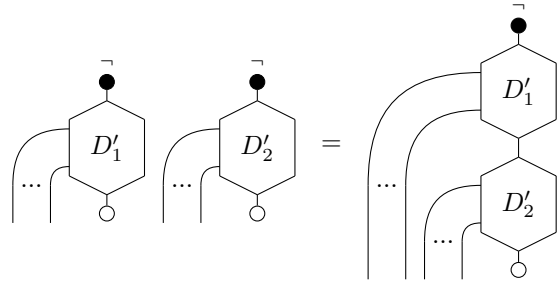
\def\fig{QMDD-tensor}
\begin{align*}
\input{./figures/\fig/\fig_00.tikz}
\eq{}\input{./figures/\fig/\fig_03.tikz}
\end{align*}
To do so, we start by factoring the top nodes:
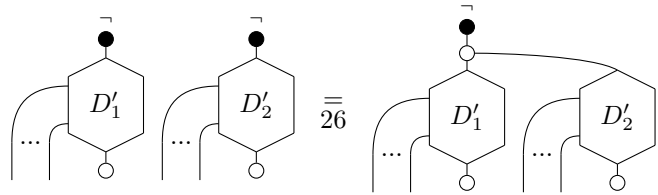
\def\fig{QMDD-tensor}
\begin{align*}
\input{./figures/\fig/\fig_00.tikz}
\eq{\ref{lem:bialgebra-neg}}\input{./figures/\fig/\fig_01.tikz}
\end{align*}
and then we push $D_2'$ in $D_1'$, which we can do layer by layer:
\def\fig{QMDD-tensor-layer}
\begin{align*}
\input{./figures/\fig/\fig_00.tikz}
\eq{\ref{lem:bialgebra-monoid-gn}}\input{./figures/\fig/\fig_01.tikz}
\eq{\ref{lem:gn-through-gadget}}\input{./figures/\fig/\fig_02.tikz}
\end{align*}
We eventually end up with:
\def\fig{QMDD-tensor}
\begin{align*}
\input{./figures/\fig/\fig_02.tikz}
\eq{(zs)\\(id)}\input{./figures/\fig/\fig_03.tikz}
\end{align*}
\end{proof}

\begin{multicols}{3}

\begin{lemma}
\label{lem:monoid-sum}
\def\fig{monoid-sums}
\begin{align*}
\input{./figures/\fig/\fig_00.tikz}
\eq{}\input{./figures/\fig/\fig_04.tikz}
\end{align*}
\end{lemma}

\begin{lemma}
\label{lem:ctrl-2-on-0-box}
\def\fig{ctrl-2-on-triangle-H}
\begin{align*}
\input{./figures/\fig/\fig_00.tikz}
\eq{}\input{./figures/\fig/\fig_07.tikz}
\end{align*}
\end{lemma}

\begin{lemma}
\label{lem:gadget-alternative}
\def\fig{gadget-alternative}
\begin{align*}
\input{./figures/\fig/\fig_00.tikz}
\eq{}\frac14\input{./figures/\fig/\fig_05.tikz}
\end{align*}
\end{lemma}

\begin{lemma}
\label{lem:projector-destroys-H-box}
\def\fig{gn-1-0-on-NF-destroy-H-spider}
\begin{align*}
\input{./figures/\fig/\fig_00.tikz}
\eq{}\input{./figures/\fig/\fig_04.tikz}
\end{align*}
\end{lemma}

\begin{lemma}
\label{lem:projector-remove-not-edge}
\def\fig{gn-1-0-on-NF-remove-not-edge}
\begin{align*}
\input{./figures/\fig/\fig_00.tikz}
\eq{}\input{./figures/\fig/\fig_03.tikz}
\end{align*}
\end{lemma}

\begin{lemma}
\label{lem:average-a-1}
\def\fig{average-a-1}
\begin{align*}
\input{./figures/\fig/\fig_00.tikz}
\eq{}2\input{./figures/\fig/\fig_03.tikz}
\end{align*}
\end{lemma}

\begin{lemma}
\label{lem:gn-1-0-on-NF-disconnect}
\def\fig{gn-1-0-on-NF-disconnect}
\begin{align*}
\input{./figures/\fig/\fig_00.tikz}
\eq{}2\input{./figures/\fig/\fig_07.tikz}
\end{align*}
\end{lemma}

\begin{lemma}
\label{lem:comonoid-alternative}
\def\fig{comonoid-alternative-2}
\begin{align*}
\frac12\input{./figures/\fig/\fig_00.tikz}
\eq{}\input{./figures/\fig/\fig_05.tikz}
\end{align*}
\end{lemma}

\end{multicols}

\begin{multicols}{2}

\begin{lemma}
\label{lem:swap-gadgets}
\def\fig{swapped-gadgets}
$$\input{./figures/\fig/\fig_00.tikz}\eq{}\input{./figures/\fig/\fig_05.tikz}$$
\end{lemma}

\begin{lemma}
\label{lem:gn-2-1-gadgets}
$$\def\fig{gn-2-1-on-gadgets-control}\input{./figures/\fig/\fig_00.tikz}\eq{}\def\fig{gn-2-1-on-gadgets-control-right}\input{./figures/\fig/\fig_00.tikz}\eq{}\input{./figures/\fig/\fig_04.tikz}$$
\end{lemma}

\end{multicols}

\begin{proof}[Proof of Lemmas \ref{lem:monoid-sum}, \ref{lem:ctrl-2-on-0-box}, \ref{lem:gadget-alternative}, \ref{lem:projector-destroys-H-box}, \ref{lem:projector-remove-not-edge}, \ref{lem:average-a-1}, \ref{lem:gn-1-0-on-NF-disconnect}, \ref{lem:comonoid-alternative}, \ref{lem:swap-gadgets}, \ref{lem:gn-2-1-gadgets}]
~

$\bullet$ \ref{lem:monoid-sum}: \cite[Prop. 5.10]{backens2021ZHcompleteness}

$\bullet$ \ref{lem:ctrl-2-on-0-box}:
\def\fig{ctrl-2-on-triangle-H}
\begin{align*}
\input{./figures/\fig/\fig_00.tikz}
\eq{\ref{lem:2-box}}\frac14\input{./figures/\fig/\fig_01.tikz}
\eq{\ref{lem:bialgebra-neg}}\frac14\input{./figures/\fig/\fig_02.tikz}
\eq{(o)}\frac12\input{./figures/\fig/\fig_03.tikz}
\eq{\ref{lem:Not-through-H-box}}\frac12\input{./figures/\fig/\fig_04.tikz}\\
\eq{(ba_2)}\frac12\input{./figures/\fig/\fig_05.tikz}
\eq{(ba_1)\\(zs)\\(id)}\frac12\input{./figures/\fig/\fig_06.tikz}
\eq{\ref{lem:0-box}}\input{./figures/\fig/\fig_07.tikz}
\end{align*}

$\bullet$ \ref{lem:gadget-alternative}:
\def\fig{gadget-alternative}
\begin{align*}
\input{./figures/\fig/\fig_00.tikz}
\eq{}\frac12\input{./figures/\fig/\fig_01.tikz}
\eq{\ref{lem:bialgebra-Z-AND}}\frac14\input{./figures/\fig/\fig_02.tikz}
\eq{}\frac14\input{./figures/\fig/\fig_03.tikz}\\
\eq{\ref{lem:Not-through-H-box}}\frac14\input{./figures/\fig/\fig_04.tikz}
\eq{(id)}\frac14\input{./figures/\fig/\fig_05.tikz}
\end{align*}

$\bullet$ \ref{lem:projector-destroys-H-box}:
\def\fig{gn-1-0-on-NF-destroy-H-spider}
\begin{align*}
\input{./figures/\fig/\fig_00.tikz}
\eq{\ref{lem:0-box}}\frac12\input{./figures/\fig/\fig_01.tikz}
\eq{(hs)\\(ba_2)}\frac14\input{./figures/\fig/\fig_02.tikz}
\eq{(ba_1)}\frac12\input{./figures/\fig/\fig_03.tikz}
\eq{\ref{lem:classical-on-H-box}\\\ref{lem:0-box}}\input{./figures/\fig/\fig_04.tikz}
\end{align*}

$\bullet$ \ref{lem:projector-remove-not-edge}:
\def\fig{gn-1-0-on-NF-remove-not-edge}
\begin{align*}
\input{./figures/\fig/\fig_00.tikz}
\eq{(hs)\\\ref{lem:Not-through-H-box}}\frac12\input{./figures/\fig/\fig_01.tikz}
\eq{\ref{lem:projector-destroys-H-box}\\(zs)}\frac12\input{./figures/\fig/\fig_02.tikz}
\eq{(id)\\(hs)}\input{./figures/\fig/\fig_03.tikz}
\end{align*}

$\bullet$ \ref{lem:average-a-1}:
\def\fig{average-a-1}
\begin{align*}
\input{./figures/\fig/\fig_00.tikz}
\eq{\ref{lem:bialgebra-neg}\\(zs)}\input{./figures/\fig/\fig_01.tikz}
\eq{\ref{lem:1-box}}\input{./figures/\fig/\fig_02.tikz}
\eq{(a)}2\input{./figures/\fig/\fig_03.tikz}
\end{align*}

$\bullet$ \ref{lem:gn-1-0-on-NF-disconnect}: First we show by induction on the number $n$ of input wires that:
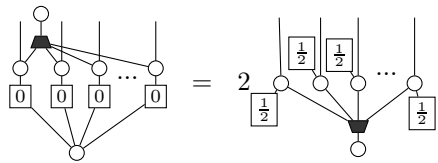
\def\fig{gn-1-0-on-NF-0-case-n}$$\input{./figures/\fig/\fig_00.tikz}\eq{}2\input{./figures/\fig/\fig_04.tikz}$$
The case $n=1$ is given by:
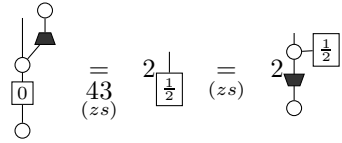
\def\fig{gn-1-0-on-NF-0-case-1}
\begin{align*}
\input{./figures/\fig/\fig_00.tikz}
\eq{\ref{lem:average-a-1}\\(zs)}2\input{./figures/\fig/\fig_01.tikz}
\eq{(zs)}2\input{./figures/\fig/\fig_02.tikz}
\end{align*}
and the general case by:
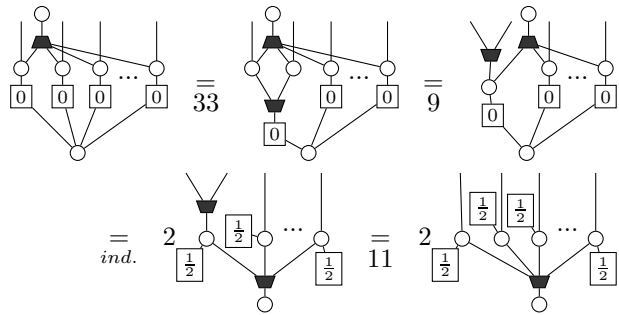
\def\fig{gn-1-0-on-NF-0-case-n}
\begin{align*}
\input{./figures/\fig/\fig_00.tikz}
\eq{\ref{lem:0-box-under-monoid}}\input{./figures/\fig/\fig_01.tikz}
\eq{\ref{lem:bialgebra-monoid-gn}}\input{./figures/\fig/\fig_02.tikz}\\
\eq{ind.}2\input{./figures/\fig/\fig_03.tikz}
\eq{\ref{lem:diagonal-distribution-monoid}}2\input{./figures/\fig/\fig_04.tikz}
\end{align*}
Finally, we show the announced equality:
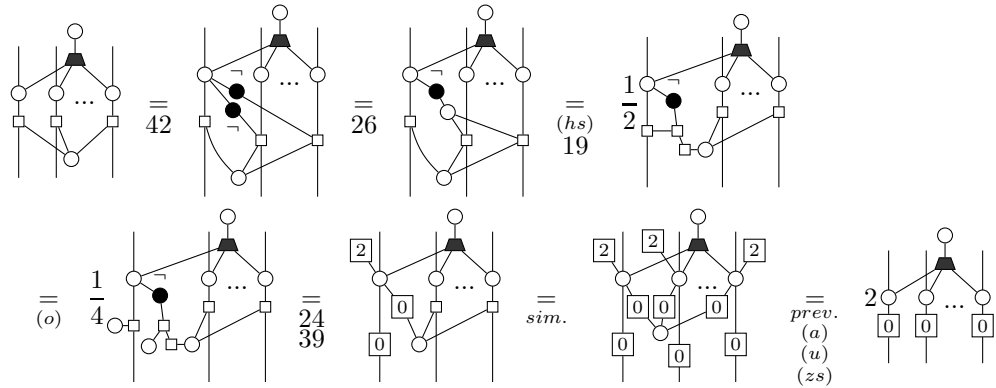
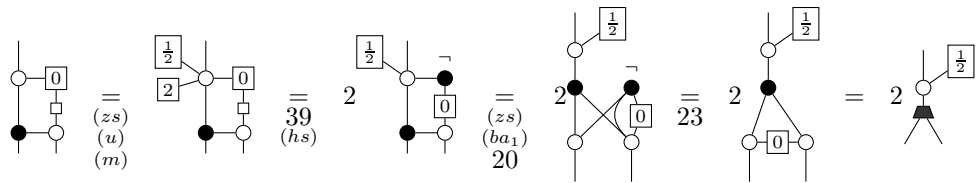
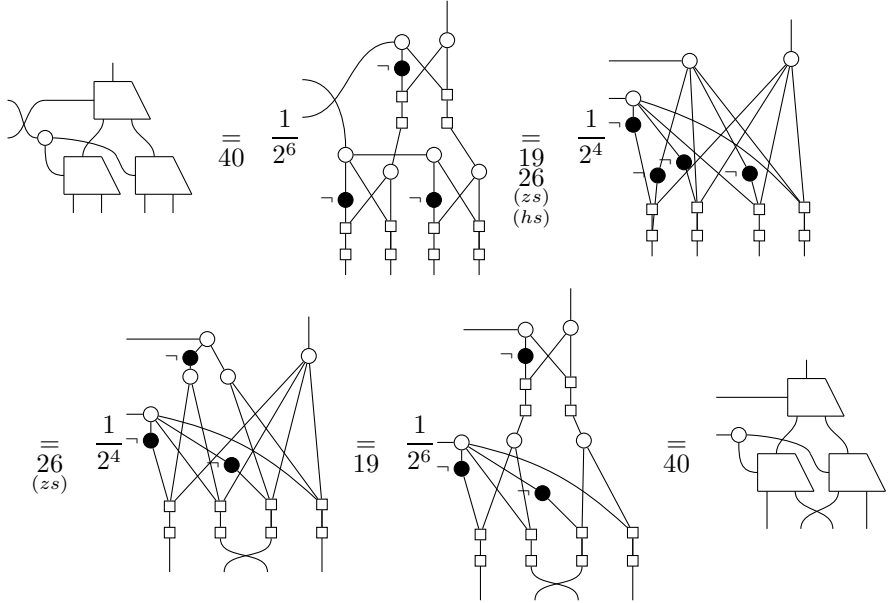
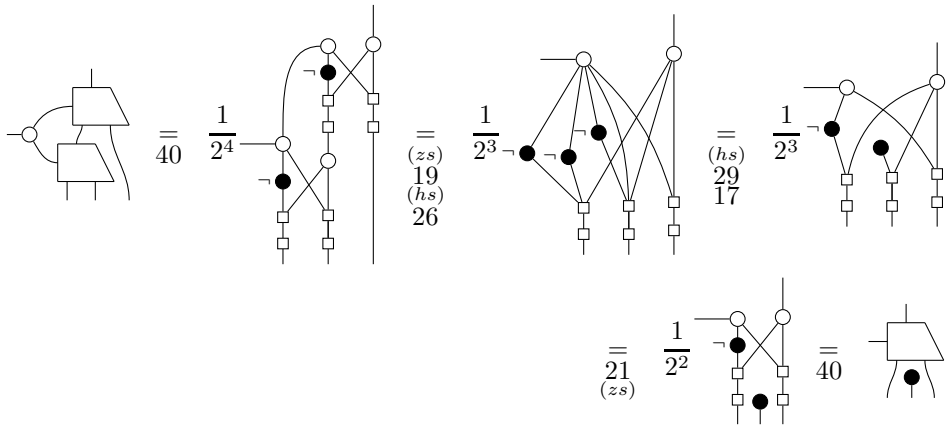
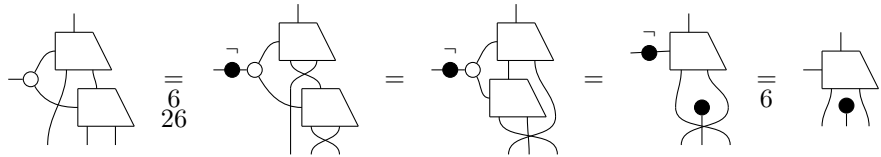
\def\fig{gn-1-0-on-NF-disconnect}
\begin{align*}
&\input{./figures/\fig/\fig_00.tikz}
\eq{\ref{lem:projector-remove-not-edge}}\input{./figures/\fig/\fig_01.tikz}
\eq{\ref{lem:bialgebra-neg}}\input{./figures/\fig/\fig_02.tikz}
\eq{(hs)\\\ref{lem:bialgebra-Z-AND}}\frac12\input{./figures/\fig/\fig_03.tikz}\\
&\eq{(o)}\frac14\input{./figures/\fig/\fig_04.tikz}
\eq{\ref{lem:0-box}\\\ref{lem:ctrl-2-on-0-box}}\input{./figures/\fig/\fig_05.tikz}
\eq{sim.}\input{./figures/\fig/\fig_06.tikz}
\eq{prev.\\(a)\\(u)\\(zs)}2\input{./figures/\fig/\fig_07.tikz}
\end{align*}

$\bullet$ \ref{lem:comonoid-alternative}:
\def\fig{comonoid-alternative-2}
\begin{align*}
\input{./figures/\fig/\fig_00.tikz}
\eq{(zs)\\(u)\\(m)}\input{./figures/\fig/\fig_01.tikz}
\eq{\ref{lem:ctrl-2-on-0-box}\\(hs)}2\input{./figures/\fig/\fig_02.tikz}
\eq{(zs)\\(ba_1)\\\ref{lem:X-spider}}2\input{./figures/\fig/\fig_03.tikz}
\eq{\ref{lem:CNOT-on-H-legs}}2\input{./figures/\fig/\fig_04.tikz}
\eq{}2\input{./figures/\fig/\fig_05.tikz}
\end{align*}

$\bullet$ \ref{lem:swap-gadgets}:
\def\fig{swapped-gadgets}
\begin{align*}
&\input{./figures/\fig/\fig_00.tikz}
\eq{\ref{lem:gadget-alternative}}\frac1{2^6}\input{./figures/\fig/\fig_01.tikz}
\eq{\ref{lem:bialgebra-Z-AND}\\\ref{lem:bialgebra-neg}\\(zs)\\(hs)}\frac1{2^4}\input{./figures/\fig/\fig_02.tikz}\\
&\eq{\ref{lem:bialgebra-neg}\\(zs)}\frac1{2^4}\input{./figures/\fig/\fig_03.tikz}
\eq{\ref{lem:bialgebra-Z-AND}}\frac1{2^6}\input{./figures/\fig/\fig_04.tikz}
\eq{\ref{lem:gadget-alternative}}\input{./figures/\fig/\fig_05.tikz}
\end{align*}

$\bullet$ \ref{lem:gn-2-1-gadgets}: First:
\def\fig{gn-2-1-on-gadgets-control-bis}
\begin{align*}
\input{./figures/\fig/\fig_00.tikz}
\eq{\ref{lem:gadget-alternative}}\frac1{2^4}\input{./figures/\fig/\fig_01.tikz}
\eq{(zs)\\\ref{lem:bialgebra-Z-AND}\\(hs)\\\ref{lem:bialgebra-neg}}\frac1{2^3}\input{./figures/\fig/\fig_02.tikz}
\eq{(hs)\\\ref{lem:hopf-gn-H}\\\ref{lem:Z-H-multiple-links}}\frac1{2^3}\input{./figures/\fig/\fig_03.tikz}\\
\eq{\ref{lem:classical-on-H-box}\\(zs)}\frac1{2^2}\input{./figures/\fig/\fig_04.tikz}
\eq{\ref{lem:gadget-alternative}}\input{./figures/\fig/\fig_05.tikz}
\end{align*}
Then:
\def\fig{gn-2-1-on-gadgets-control-right}
\begin{align*}
\input{./figures/\fig/\fig_00.tikz}
\eq{\ref{lem:swapped-gadget-legs}\\\ref{lem:bialgebra-neg}}\input{./figures/\fig/\fig_01.tikz}
\eq{}\input{./figures/\fig/\fig_02.tikz}
\eq{}\input{./figures/\fig/\fig_03.tikz}
\eq{\ref{lem:swapped-gadget-legs}}\input{./figures/\fig/\fig_04.tikz}
\end{align*}

\end{proof}

\begin{proof}[Proof of Proposition \ref{prop:SQMDD-swap}]
\phantomsection\label{prf:SQMDD-swap}
We use Lemma \ref{lem:swap-gadgets} to show how two layers interact when a swap is applied to their corresponding outputs:
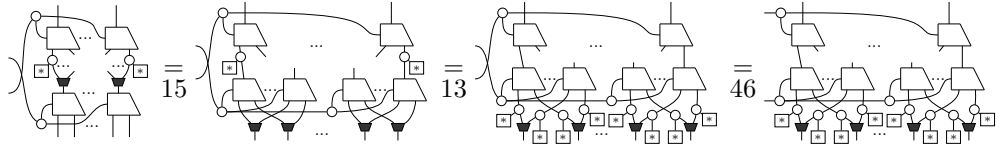
\def\fig{swapped-SQMDD-legs}
\begin{align*}
\scalebox{0.65}{\input{./figures/\fig/\fig_00.tikz}}
\eq[]{\ref{lem:bialgebra-monoid-gadget}}\scalebox{0.65}{\input{./figures/\fig/\fig_01.tikz}}
\eq[]{\ref{lem:diagonal-distribution-gadget}}\scalebox{0.65}{\input{./figures/\fig/\fig_02.tikz}}
\eq[]{\ref{lem:swap-gadgets}}\scalebox{0.65}{\input{./figures/\fig/\fig_03.tikz}}
\end{align*}
The resulting diagram is in SQMDD form, and can be reduced further using the above strategy.
\end{proof}

\begin{proof}[Proof of Proposition \ref{prop:SQMDD-gn-2-1}]
\phantomsection\label{prf:SQMDD-gn-2-1}
If $\tikzfig{gn-2-1}$ is applied to two arbitrary outputs of $D$, using Proposition \ref{prop:SQMDD-swap}, we can swap the outputs so that $\tikzfig{gn-2-1}$ is now applied on the first two.

We are left with two cases, the first one being when the ``root gadget'' is linked to the same ``child gadget'', in which case we apply Lemma \ref{lem:bialgebra-monoid-gadget}:
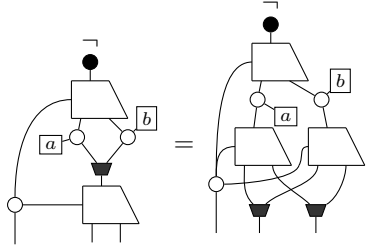
$$\tikzfig{gn-2-1-on-gadgets-control-pre-fin}$$
The second case now encapsulates the first, and is dealt with as follows:
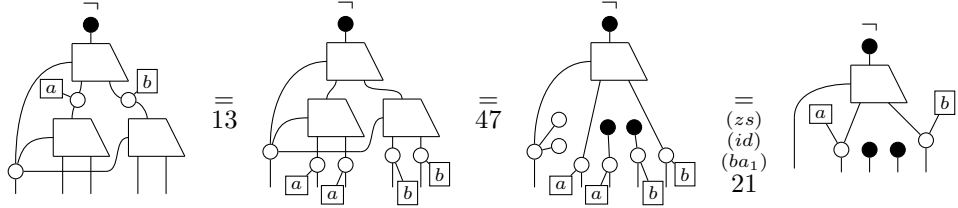
\def\fig{gn-2-1-on-gadgets-control-fin}
\begin{align*}
\input{./figures/\fig/\fig_00.tikz}
\eq{\ref{lem:diagonal-distribution-gadget}}\input{./figures/\fig/\fig_01.tikz}
\eq{\ref{lem:gn-2-1-gadgets}}\input{./figures/\fig/\fig_02.tikz}
\eq{(zs)\\(id)\\(ba_1)\\\ref{lem:classical-on-H-box}}\input{./figures/\fig/\fig_03.tikz}
\end{align*}
Now using the rewrite strategy of Proposition \ref{prop:preserved-semantics}, the nodes $\tikzfig{ket-0}$ will trickle down the diagram, eventually getting absorbed by \def\fig{monoid-def}$\input{./figures/\fig/\fig_00.tikz}$, and potentially destroying some gadgets as they do.
\end{proof}

\begin{proof}[Proof of Proposition \ref{prop:SQMDD-gn-1-0}]
\phantomsection\label{prf:SQMDD-gn-1-0}
Using Proposition \ref{prop:SQMDD-swap}, we may consider we only apply $\tikzfig{bra-plus}$ on the very last output. We notice that:
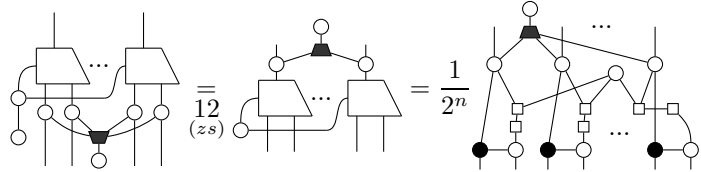
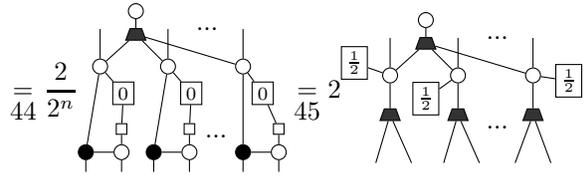
\def\fig{gn-1-0-on-NF-to-comonoid}
\begin{align*}
\input{./figures/\fig/\fig_00.tikz}
&\eq[]{\ref{lem:gn-through-gadget}\\(zs)}\input{./figures/\fig/\fig_01.tikz}
\eq[]{}\frac1{2^n}\input{./figures/\fig/\fig_02.tikz}\\
&\eq[]{\ref{lem:gn-1-0-on-NF-disconnect}}\frac2{2^n}\input{./figures/\fig/\fig_03.tikz}
\eq[]{\ref{lem:comonoid-alternative}}2\input{./figures/\fig/\fig_04.tikz}
\end{align*}
With this result, we can show how we turn the diagram into SQMDD form:
\def\fig{gn-1-0-on-NF-final-bis}
\begin{align*}
\input{./figures/\fig/\fig_00.tikz}
\eq{}2\input{./figures/\fig/\fig_01.tikz}
\eq{\ref{lem:monoid-sum}}2\input{./figures/\fig/\fig_02.tikz}
\end{align*}

\end{proof}

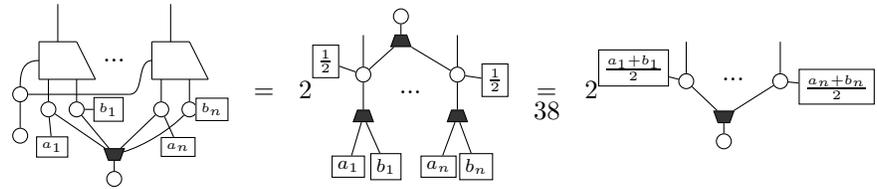
\end{document}

%% file: figures/monoid-sums/monoid-sums_00.tikz
\begin{tikzpicture}
	\begin{pgfonlayer}{nodelayer}
		\node [style=none] (13)  at (0.0, -0.438) {};
		\node [style=comonoid, shape border rotate=180] (14)  at (0.0, 0.062) {};
		\node [style=box] (15)  at (-0.25, 0.438) {$a$};
		\node [style=box] (16)  at (0.25, 0.438) {$b$};
	\end{pgfonlayer}
	\begin{pgfonlayer}{edgelayer}
		\draw (14) to (13.center);
		\draw (14) to (16);
		\draw (15) to (14);
	\end{pgfonlayer}
\end{tikzpicture}

%% file: figures/monoid-sums/monoid-sums_01.tikz
\begin{tikzpicture}
	\begin{pgfonlayer}{nodelayer}
		\node [style=white dot] (42)  at (0.05, -0.337) {};
		\node [style=none] (43)  at (0.05, -0.863) {};
		\node [style=box] (44)  at (-0.275, 0.863) {$a$};
		\node [style=box] (45)  at (0.375, 0.863) {$b$};
		\node [style=box] (46)  at (0.375, 0.163) {};
		\node [style=box] (47)  at (0.375, 0.512) {};
		\node [style=box] (48)  at (-0.275, 0.163) {};
		\node [style=box] (49)  at (-0.275, 0.512) {};
		\node [style=dot] (50)  at (0.05, 0.163) {};
		\node [style=none, font={\scriptsize}] (51)  at (0.05, 0.413) {$\neg$};
		\node [style=box] (53)  at (-0.375, -0.462) {$2$};
	\end{pgfonlayer}
	\begin{pgfonlayer}{edgelayer}
		\draw (42) to (48);
		\draw (42) to (46);
		\draw (42) to (43.center);
		\draw (44) to (48);
		\draw (45) to (46);
		\draw (46) to (47);
		\draw (48) to (46);
		\draw (53) to (42);
	\end{pgfonlayer}
\end{tikzpicture}

%% file: figures/monoid-sums/monoid-sums_02.tikz
\begin{tikzpicture}
	\begin{pgfonlayer}{nodelayer}
		\node [style=white dot] (54)  at (0.05, -0.362) {};
		\node [style=none] (55)  at (0.05, -0.887) {};
		\node [style=box] (58)  at (0.375, 0.137) {$b$};
		\node [style=box] (60)  at (-0.275, 0.137) {$a$};
		\node [style=dot] (62)  at (0.05, 0.637) {};
		\node [style=none, font={\scriptsize}] (63)  at (0.05, 0.887) {$\neg$};
		\node [style=box] (65)  at (-0.375, -0.487) {$2$};
	\end{pgfonlayer}
	\begin{pgfonlayer}{edgelayer}
		\draw (54) to (60);
		\draw (54) to (58);
		\draw (54) to (55.center);
		\draw (62) to (60);
		\draw (62) to (58);
		\draw (65) to (54);
	\end{pgfonlayer}
\end{tikzpicture}

%% file: figures/monoid-sums/monoid-sums_04.tikz
\begin{tikzpicture}
	\begin{pgfonlayer}{nodelayer}
		\node [style=none] (75)  at (0.0, -0.263) {};
		\node [style=box] (76)  at (0.0, 0.263) {$a{+}b$};
	\end{pgfonlayer}
	\begin{pgfonlayer}{edgelayer}
		\draw (76) to (75.center);
	\end{pgfonlayer}
\end{tikzpicture}

%% file: figures/monoid-sums/monoid-sums_03.tikz
\begin{tikzpicture}
	\begin{pgfonlayer}{nodelayer}
		\node [style=white dot] (66)  at (0.125, -0.112) {};
		\node [style=none] (67)  at (0.125, -0.637) {};
		\node [style=box] (69)  at (0.3, 0.637) {$\frac{a+b}2$};
		\node [style=box] (73)  at (-0.3, -0.237) {$2$};
	\end{pgfonlayer}
	\begin{pgfonlayer}{edgelayer}
		\draw (66) to (69);
		\draw (66) to (67.center);
		\draw (73) to (66);
	\end{pgfonlayer}
\end{tikzpicture}